\documentclass[a4paper,numberwithinsect,cleveref,english]{lipics-v2021}
\usepackage{xspace,amsmath,amsfonts,amssymb,bbm,ifthen,mathtools}

\usepackage{tikz}
\usetikzlibrary{arrows,automata,shapes,calc,decorations.markings,positioning,patterns.meta,intersections,decorations.pathreplacing}
\tikzset{
  myautomata/.style={%
    ->, >=stealth', auto, initial text=,}
  state/.style={%
    inner sep=.5mm, minimum size=3.5mm, draw=black, circle
  }
  state with output/.style={%
    shape=rectangle split, rectangle split parts=2, draw, fill=white,
    initial text=, inner sep=1mm
  }
}

\nolinenumbers

\usepackage[
    type={CC},
    modifier={by},
    version={4.0},
]{doclicense}

\usepackage[Tol]{colorblind}
\usepackage{bm}

\usepackage[noend]{algpseudocode}
\usepackage[bibliography=common]{apxproof}

\AddToHook{env/proposition/begin}{\crefalias{theorem}{proposition}}
\AddToHook{env/lemma/begin}{\crefalias{theorem}{lemma}}
\AddToHook{env/example/begin}{\crefalias{theorem}{example}}
\AddToHook{env/definition/begin}{\crefalias{theorem}{definition}}
\AddToHook{env/corollary/begin}{\crefalias{theorem}{corollary}}
\AddToHook{env/claim/begin}{\crefalias{theorem}{claim}}
\AddToHook{env/remark/begin}{\crefalias{theorem}{remark}}

\usepackage{algorithm}
\usepackage[indLines=true,noEnd=true,beginComment=/*~,endComment=~*/,
  beginLComment=/*~, endLComment=~*/, commentColor=dark gray]{algpseudocodex}
\tikzset{algpxIndentLine/.style={draw=gray,very thin,-}}
\usepackage{caption}

\renewcommand{\algorithmicrequire}{\textbf{Input:}}
\renewcommand{\algorithmicensure}{\textbf{Output:}}

\newtheoremrep{theorem}{Theorem}[section]       %
\newtheoremrep{proposition}[theorem]{Proposition}
\newtheoremrep{lemma}[theorem]{Lemma}
\newtheoremrep{example}[theorem]{Example}
\newtheoremrep{definition}[theorem]{Definition}
\newtheoremrep{corollary}[theorem]{Corollary}
\newtheoremrep{conjecture}[theorem]{Conjecture}
\newtheoremrep{claim}[theorem]{Claim}
\newtheoremrep{remark}[theorem]{Remark}
\newtheoremrep{fact}[theorem]{Fact}

\newcommand{\ZG}{ZG\xspace}
\newcommand{\LZG}{LZG\xspace}
\newcommand{\QLZG}{QLZG\xspace}
\newcommand{\threshext}{semi-extensible \ZG}
\newcommand{\Threshext}{Semi-extensible \ZG}
\newcommand{\THreshext}{Semi-Extensible \ZG}
\newcommand{\cond}{\text{Cond}}
\newcommand{\allrare}{case \emph{all-rare}\xspace}
\newcommand{\nocond}{case \emph{no-cond}\xspace}

\newcommand{\NN}{\mathbb{N}}

\title{Constant-Time Dynamic Enumeration of Word Infixes in a Regular Language}

\author{Antoine Amarilli}{Univ. Lille, Inria, CNRS, Centrale Lille, UMR 9189 CRIStAL, F-59000 Lille, France \and \url{https://a3nm.net/}}{a3nm@a3nm.net}{https://orcid.org/0000-0002-7977-4441}{
Partially supported by the ANR project EQUUS ANR-19-CE48-0019 and by the
Deutsche Forschungsgemeinschaft (DFG, German Research Foundation) – 431183758.}
\author{Sven Dziadek}{SAMOVAR, Télécom SudParis, Institut Polytechnique de Paris, Palaiseau, France}{sven.dziadek@telecom-sudparis.eu}{https://orcid.org/0000-0001-6767-7751}{Supported by the ANR project STeP2 ANR-22-EXES-0013. Partially supported by the ANR project EQUUS ANR-19-CE48-0019 and by the
Deutsche Forschungsgemeinschaft (DFG, German Research Foundation) – 431183758.}
\author{Luc Segoufin}{Inria \& ENS Paris, France}{luc.segoufin@inria.fr}{https://orcid.org/0000-0002-9564-7581}{}

\authorrunning{A. Amarilli, S. Dziadek, and L. Segoufin} %

\Copyright{Antoine Amarilli, Sven Dziadek, and Luc Segoufin} %


\relatedversion{}

\hideLIPIcs

\acknowledgements{We would like to thank the reviewers for their valuable feedback.}

\begin{document}

\maketitle

\begin{abstract}
  For a fixed regular language $L$, the \emph{enumeration of
    $L$-infixes} is the following task: we are given an input word $w=a_1 \cdots
    a_n$ and we must enumerate the infixes of~$w$ that belong to~$L$, i.e., the
  pairs $i \leq j$ such that $a_i \cdots a_j \in L$. We are interested in
  \emph{dynamic enumeration of $L$-infixes}, where we must additionally support 
  letter substitution updates on~$w$ (e.g., ``replace the $i$-th letter of~$w$ by a letter~$a$''). Each update changes the set of infixes to enumerate, and resets the
  enumeration state.

  We study for which regular languages $L$ we can perform dynamic enumeration of
  $L$-infixes in constant delay (i.e., the next infix is always produced in
  constant time) and constant additional memory throughout the enumeration,
  while supporting each update in constant time.

  We show that, for languages $L$ with a neutral letter, if the language $L$
  belongs to the class ZG and is
  \emph{extensible} (i.e., if $u \in L$ and $u$ is a factor of $v$ then $v \in
  L$), then dynamic enumeration of $L$-infixes can be achieved with a
  simple algorithm that ensures constant-time updates and constant delay, but not
  constant additional memory. Our main contribution is then to show an algorithm
  that additionally uses only constant additional memory, and applies to a more
  general class of \emph{semi-extensible ZG} languages for which we give several
  equivalent characterizations. 
  We further discuss whether our results can be generalized to larger
  language classes and show some (conditional) lower bounds.
\end{abstract}

\section{Introduction}

Given a regular language $L$ and a word $w=a_1 \cdots a_n$ the
\emph{$L$-infixes} of $w$ are the pairs $[i,j]$ of~$w$ such that the factor
$a_i \cdots a_j$ belongs to $L$. We study in this paper how to compute all $L$-infixes of $w$
in the framework of \emph{enumeration algorithms}: we distinguish
between the \emph{preprocessing time} required to read the input and
compute the first result, the \emph{delay} between any two successive
results, and the \emph{additional memory} used throughout the enumeration
(i.e., the memory which is written during the enumeration phase).
We always assume in this paper that the regular language $L$ is
fixed: we measure complexity only as a function of the length~$n$ of the input
word, and do not study how the complexity depends on~$L$.
In this setting, the best possible efficiency guarantees for an enumeration
algorithm is to ensure that it is 
\emph{linear preprocessing and constant delay}~\cite{kazana2011first,kazana2013enumeration}, i.e., 
its preprocessing time is linear in $n$ and that its delay is independent from~$n$.

Enumeration algorithms have been studied in many settings~\cite{Wasa16} and in
particular in database theory~\cite{Segoufin14}, where they have been used to
efficiently enumerate query answers. In particular, as $L$-infixes can be
expressed in MSO, it follows from~\cite{bagan2006mso} that there are efficient
enumeration algorithms to compute $L$-infixes for any regular language~$L$.
Other enumeration algorithms have been
designed in several special cases, in particular for \emph{document
spanners}~\cite{florenzano2018constant,amarilli2019constant}, a formalism for
declarative information extraction~\cite{fagin2015document}: these algorithms 
can in particular be used to enumerate $L$-infixes while also ensuring 
that the complexity remains reasonable as a function of the size of an automaton representing~$L$.

However, one drawback of these enumeration algorithms is that
their linear-time preprocessing phase must be re-run from scratch whenever the
word is modified.  This has motivated the investigation of enumeration
algorithms on dynamic data, which aim at enumerating solutions
while supporting efficient updates to the input word. In this setting,
whenever the input word is modified, the enumeration of the infixes is
restarted from scratch. Dynamic
enumeration was first studied for MSO queries over
trees~\cite{losemann2014mso}, supporting polylogarithmic delay and update time.
Finer bounds on the evaluation of MSO over words were then shown
in~\cite{niewerth2018enumeration}: for any fixed MSO query, given a word $w$ of
length $n$, one can compute in $O(n)$ a data structure allowing constant delay
enumeration of the query results with constant auxiliary memory, and supporting $O(\log n)$ update time
(allowed updates are relabeling, insertions, and deletions).  This implies that the same
bounds hold for the dynamic enumeration of $L$-infixes for any fixed regular
language~$L$.
Dynamic enumeration was later studied in the context of
document spanners, with support for SLP-represented documents and with
complex editing operations~\cite{schmid2022query}; and also
for the specific task of pattern matching, with algorithms
that achieve efficiency also in the language $L$ describing the searched pattern~\cite{alstrup2000pattern,amir2007dynamic}.

While linear preprocessing, constant delay, and constant auxiliary memory are hard to beat, it is not
clear whether the logarithmic update time is the best possible.  It
turns out that for some regular languages, the answer is essentially yes: there
are regular languages $L$ admitting an unconditional
$\Omega(\log n / \log \log n)$ lower bound in the cell probe model for the
\emph{dynamic membership} task of maintaining whether a word belongs to~$L$
under letter substitution updates (\cite{frandsen1997}, using
results from~\cite{fredman1989cell}; see also~\cite{amarilliICALP21}).  This
can be easily transformed into a similar lower bound for the dynamic
enumeration of $L$-infixes for some regular languages~$L$. However, there are
also regular languages $L$ for which the dynamic enumeration problem can be solved
more efficiently, i.e. with linear preprocessing time, constant delay, constant
update time, and constant auxiliary memory.
Our goal in this paper is to identify the regular languages for which dynamic
enumeration for $L$-infixes enjoys this favorable complexity regime (with
letter substitution updates).

\subparagraph*{Contributions.}
We introduce a natural property on regular languages, called being
\emph{extensible}: $L$ is extensible if whenever $u \in L$ then
$sut \in L$ for all words $s$ and~$t$. 
Extensible regular languages are an interesting class of languages to study
for $L$-infix enumeration, because extensibility ensures
that any infix that extends an $L$-infix is itself an $L$-infix.
It further turns out that, for extensible
languages~$L$, the dynamic enumeration problem for $L$-infixes is at least as hard as
the dynamic membership problem for~$L$. Hence, focusing on regular languages with a
so-called \emph{neutral letter} (i.e., a letter having no influence on
membership to the language, which is a common simplifying assumption~\cite{koucky2005bounded,barrington2005first}), the good
candidates for tractable dynamic enumeration of $L$-infixes are the languages from the class \emph{ZG} studied
in~\cite{amarilliICALP21}: these are (conditionally) the only languages with
neutral letters for
which dynamic membership is feasible in constant time per update.

Our first contribution is to show that, in the presence of a neutral
letter and assuming extensibility, the tractability of ZG languages $L$
for dynamic membership extends to
dynamic enumeration of $L$-infixes.
Specifically (\cref{prp:simple}): for any fixed
extensible language~$L$ featuring a neutral letter, we can devise an algorithm
for the dynamic enumeration problem for $L$-infixes that achieves linear
preprocessing, constant delay enumeration, and
constant-time updates.
The algorithm is rather simple and proceeds using an oracle for
dynamic membership to~$L$ as a black-box, together with an even-odd
trick~\cite{uno2003two}. However, this algorithm forgoes the constant additional memory requirement.
Our main contribution is thus to propose a dynamic enumeration algorithm
that also uses only constant additional memory. We do so for extensible ZG languages, and in fact for a slight extension 
that we introduce in the present paper, called
\emph{semi-extensible ZG languages}.
Intuitively,
such languages relax the requirement of extensibility by enforcing
it only on words that contain a sufficiently large number of occurrences of some non-neutral
letter. So we show (\cref{thm:main}): for any fixed semi-extensible ZG language $L$, we can devise an algorithm
for dynamic enumeration of $L$-infixes that achieves linear
preprocessing, 
constant delay enumeration,
constant-time updates, and constant additional memory.

The key data structure involved in the proof of \cref{thm:main} is rather
simple: it consists in maintaining for each letter an unordered list of all its
occurrences in the word, called \emph{occurrence lists}. This can be done using
well-known data structures, see for instance the dynamic membership algorithms
for ZG languages in~\cite{amarilliICALP21}).  The enumeration algorithm based
on these occurrences lists is however rather technical. It performs 
two kinds of background traversal processes (which are amortized over the rest
of the algorithm) and memorization of a carefully designed constant amount of
\emph{left information} and \emph{right information} in relation to a pointer
that sweeps over the word.

We finally show that semi-extensible ZG languages are not the only languages
admitting tractable dynamic enumeration by showing other cases
with ad-hoc algorithms: e.g., the non-semi-extensible ZG language $L = (aa)^* a$,
or $L = b^* a$ and $L = a \Sigma^* a$ (which are non-extensible languages
outside of ZG). We show however that there are (non-semi-extensible) ZG
languages $L$ that do not admit dynamic enumeration of $L$-infixes with
constant-time updates and constant delay enumeration, subject to the
\emph{prefix-$U_1$} complexity hypothesis of~\cite{amarilliICALP21}. All of this
justifies that the constant-time tractability boundary for dynamic infix
enumeration is in fact incomparable with the class ZG.

\subparagraph*{Paper structure.}
We give preliminaries and formally introduce the task of dynamic infix
enumeration in \cref{sec:prelim}. We
briefly review the results of~\cite{amarilliICALP21}  in
\cref{sec:dynamic_membership}. We present our simple algorithm for extensible ZG
languages (\cref{prp:simple}) in \cref{sec:extensible}. We then define the class of
semi-extensible ZG languages in \cref{sec:threshext}. We give our
constant additional memory algorithm for semi-extensible ZG languages (\cref{thm:main}) in 
\cref{sec:main}. We discuss other tractable cases and lower bounds in
\cref{sec:lower}, and conclude in \cref{sec:conc}. 
Most detailed proofs are deferred to the appendix.

\section{Preliminaries and Problem Statement}
\label{sec:prelim}

We fix a finite alphabet $\Sigma$.  We write $\Sigma^*$ for the set of words
over~$\Sigma$, write $|w|$ for the length of a word $w \in \Sigma^*$, and write
$\epsilon$ for the empty word.  For $a\in\Sigma$ and $w \in \Sigma^*$, we write
$|w|_a$ for the number of occurrences of $a$ in~$w$.  We say that
$u \in \Sigma^*$ is a \emph{factor} of $v \in \Sigma^*$ if there exist
$s,t \in \Sigma^*$ such that $v = sut$.  A language $L$ over an alphabet
$\Sigma$ is a subset of $\Sigma^*$. A letter $e \in \Sigma$ is \emph{neutral}
for a language $L$ if, for any two words $s, t \in \Sigma^*$, we have
$st \in L$ iff $set \in L$.

\subparagraph*{Infixes.}  Writing a word $w \in \Sigma^*$ as
$w = a_1 \cdots a_n$, a \emph{position} of~$w$ is a value $i$ with
$1 \leq i \leq n$. We write $w[i]$ to mean the $i$-th letter $a_i$ of~$w$.  An
\emph{infix} of~$w$ is a 2-tuple $[i,j]$ of positions of~$w$, with
$1 \leq i \leq j \leq n$: its \emph{left endpoint} is~$i$, its \emph{right
  endpoint} is~$j$, and it \emph{realizes} the word
$w[i,j] \coloneq a_i \cdots a_j$ which is a factor of~$w$. Note that infixes
always realize nonempty factors by definition, and that several infixes may
realize the same factor.  For $L \subseteq \Sigma^*$ a language, an
\emph{$L$-infix} of~$w$ is an infix realizing a word of~$L$.

\subparagraph*{Dynamic enumeration problem.} Given a language $L$, the problem
of \emph{dynamic enumeration for $L$-infixes} takes as input a word $w$ and
must build a data structure on~$w$ that will support efficient
\emph{enumeration} of all $L$-infixes of $w$ and be efficiently maintainable
under \emph{letter substitution updates} on~$w$.  We will assume that $L$ is
fixed, so that complexity will only be measured as a function of the length
of~$w$ (which is never changed by updates).

Formally, we consider data structures which must support two operations:
\begin{itemize}
\item \emph{Letter substitution updates} on the current word $w$: the update is
  specified by a
  position $1 \leq i \leq |w|$ and a letter $a \in \Sigma$, and it modifies
  $w$ by replacing its $i$-th character by~$a$.
\item \emph{Enumeration queries}: the query asks us to output all $L$-infixes of the current word
  exactly once. We do not specify in which order the $L$-infixes are
  returned.
\end{itemize}

\begin{example}
  Consider the language $L = a^*$ and the initial word $w = aaa$. An enumeration
  query produces the six infixes $[1,1]$, $[1,2]$, $[1,3]$,
  $[2,2]$, $[2,3]$ and $[3,3]$ (in some order).  An update operation $(2,b)$
  changes $w$ to $aba$. If we now perform again an enumeration
  query, we receive the two infixes $[1,1]$ and $[3,3]$.  
  We can then perform two successive update operations $(1,b)$ and $(3,b)$, changing the
  word to $w = bbb$. If we now perform an enumeration query, we
  receive no results. 
\end{example}

The goal of this paper is to design data structures for the dynamic enumeration
problem for $L$-infixes which efficiently support letter substitution updates
and enumeration queries.  Our focus is on achieving \emph{constant-time
  updates}, and \emph{constant delay} and \emph{constant memory enumeration},
where \emph{constant} means independent from the length of the word~$w$. We
also require \emph{linear preprocessing and linear memory}.  Formally, we mean
that there is an integer $B \in \mathbb{N}$, depending only on~$L$, such that,
for any word~$w$:

\begin{description}
  \item[Constant-time updates.] Each update is processed in at most $B$
    computation steps.
  \item[Constant delay enumeration.] While processing an enumeration query,
    after at most $B$ computation steps, we produce the next $L$-infix (or
    conclude that none are left).
  \item[Constant memory enumeration.] During the processing of an enumeration query, the total
    additional number of memory cells used is bounded by~$B$. (The state of the memory before we started the
    enumeration query can still be used but in a read-only fashion.)

    \item[Linear preprocessing and linear memory.] The running time to build the
      data structure on the initial word is bounded by $B |w|$ steps, and the total number of memory
      cells used by the data structure is always bounded by $B |w|$ (i.e., it
      never exceeds this during the lifetime of the structure, no matter how
      many updates or enumeration queries are performed).
\end{description}

We use the unit-cost RAM model~\cite{grandjean2022which}, with cell size
logarithmic in the word length~$n \coloneq |w|$ (which never changes). For
instance, an integer in range $\{1, \ldots, n\}$ fits in one memory cell.

Note that the constant memory enumeration constraint implies that we can
freely interrupt the enumeration process when an update arrives. Indeed, it ensures
that the data structure used to support the updates is not modified during the 
enumeration phase, so that an incoming update can be applied at any time.
However, after each update, the enumeration status is reset: 
enumeration queries always restart from scratch after an update, which is
sensible because the set of $L$-infixes to enumerate may have changed.

In this paper we study the problem of dynamic enumeration of $L$-infixes for
regular languages~$L$.  We show language families for which we can construct
data structures achieving the requirements above, and languages for which this
is not possible (sometimes conditionally).

\section{Dynamic Membership and \ZG Languages}
\label{sec:dynamic_membership}
\begin{toappendix}
  \label{apx:zg}
\end{toappendix}

The dynamic enumeration problem for $L$-infixes is related to the task of
maintaining membership to $L$ under updates, which is called \emph{dynamic
  membership} and is studied in~\cite{frandsen1997,amarilliICALP21}. In this
section, we review the definition of the dynamic membership problem and some
results from~\cite{frandsen1997,amarilliICALP21}, focusing on a class of
regular languages called \emph{\ZG}. We recall the definition of the
\emph{syntactic monoid} of a language and define \emph{\ZG languages} as those
languages whose syntactic monoid falls in a specific class of monoids, also
called \ZG. We then state the result from~\cite{frandsen1997,amarilliICALP21}
that \ZG languages enjoy constant-time dynamic membership, and that this result
is tight for languages having a neutral letter, conditionally to a complexity
hypothesis.

\subparagraph*{Dynamic membership.}  Consider a language $L$. The \emph{dynamic
  membership} problem is similar to the dynamic enumeration problem mentioned
in the previous section, and asks us to design a data structure that can handle
the same letter substitution updates. However, enumeration queries are replaced
by a membership query that asks whether the current word is in $L$.

It was shown in~\cite{amarilliICALP21} that, for a certain class \QLZG of
regular languages, there is a data structure with linear preprocessing and
linear memory that can \emph{solve dynamic membership in constant time}, i.e., both
letter substitution updates and membership queries are supported in
constant time.
This is analogous to our goal for dynamic enumeration of $L$-infixes.

Further, it is shown in~\cite{amarilliICALP21} that \QLZG is the largest class
of regular languages having this property, conditionally to a computational
assumption called \emph{prefix-$U_1$} that we will review in \cref{sec:lower}.
The definition of the class \QLZG is somewhat technical, so we will review a
simpler subclass of \QLZG from~\cite{amarilliICALP21}, called \ZG, that is
sufficient for our needs. The two classes coincide for languages with a neutral
letter, which are the main focus of this paper.

\subparagraph*{\ZG languages and \ZG monoids.}  The \emph{syntactic monoid}~$M$
of a regular language~$L$ is the monoid obtained by quotienting the words
of~$\Sigma^*$ by the syntactic equivalence relation $\sim_L$ defined as
follows: we have $u \sim_L v$ for $u, v \in \Sigma^*$ whenever we have
$sut \in L$ iff $svt \in L$ for each $s, t \in \Sigma^*$. As the syntactic
equivalence relation $\sim_L$ is a congruence (i.e., when $u \sim_L v$ and
$u' \sim_L v'$ then $uu' \sim_L vv'$), then the concatenation of words gives a
composition law on the equivalence classes of~$\Sigma^*$ under~$\sim_L$.  In
particular, any neutral letter is syntactically equivalent to the empty
word~$\epsilon$, which corresponds to the neutral element of~$M$.  It is
well-known that the syntactic monoid of a regular language is always finite,
and that any finite monoid $M$ has an \emph{idempotent power}, denoted
$\omega$, such that $x^\omega=x^\omega x^\omega$ for all $x \in M$.

The class of monoids \emph{\ZG} consists of all finite monoids $M$ satisfying
the equation $yx^{\omega+1}=x^{\omega+1}y$ for every $x, y \in M$.
The class of languages \emph{\ZG} then consists of
those languages whose syntactic monoid is in \ZG.
Intuitively, \ZG languages are those languages where frequent letters
are \emph{central}, i.e., they commute with all other elements.

\begin{example}
  Examples of \ZG languages are $L_1=(aa)^*$, $L_2=ab$, $L_3 = \{w\in \{a,b\}^* \mid|w|_a\geq 3 \text{ and } |w|_b
  \mod 3 = 0\}$,
$L_4= (a+b)^* c
  (a+b)^* d (a+b)^*$, 
  and $L_5 = L_3 \cap L_4$.
\end{example}

\subparagraph*{Tractability for \ZG languages.}
Let us summarize in the following proposition what is known
from~\cite{amarilliICALP21} about dynamic membership with constant time per
update:

\begin{propositionrep}[Follows from Theorem~1.1 in \cite{amarilliICALP21}, see
  Appendix~\ref{apx:zg}]
  \label{prp:zgupdates}
  Let $L$ be a regular language. If $L$ is in \ZG, then the dynamic membership
  problem to~$L$ can be solved in constant time with linear preprocessing and
  linear memory. Conversely, assuming the prefix\nobreakdash-$U_1$ hypothesis
  (see \cref{sec:lower}), and assuming
  that $L$ features a neutral letter,
  if the dynamic membership problem to~$L$ can be solved
  in constant time, then $L$ is in \ZG.
\end{propositionrep}

\begin{toappendix}
  We show in the rest of this appendix how \cref{prp:zgupdates}
  follows from Theorem~1.1 in~\cite{amarilliICALP21}, which shows the following:

  \begin{theorem}[Special case of Theorem~1.1 in~\cite{amarilliICALP21}]
    Let $L$ be a regular language.
  \begin{itemize}
    \item If $L$ is in \QLZG then the dynamic membership
  problem to~$L$ can be solved in constant time;
\item Assuming the prefix-$U_1$ hypothesis,
  if the dynamic membership problem to~$L$ can be solved
  in constant time, then $L$ is in \QLZG.
  \end{itemize}
  \end{theorem}

  So, let $L$ be a regular language, and 
  let $\Sigma$ be the alphabet over which $L$ is defined.
  Recall that the \emph{syntactic semigroup} $S$ of~$L$ is the
  semigroup obtained by quotienting the \emph{non-empty} words of~$\Sigma^+$ by
  the syntactic equivalence relation. 
  This definition is the analogue of the syntactic monoid~$M$, except that we
  disallow the empty word: $M$ is equal to $S$
  whenever $S$ contains a neutral element (which is in
  particular true if $L$ features a neutral letter), and otherwise $M$
  is equal to $S$ up to adding an identity element.
  Further, the \emph{syntactic morphism} $\eta$ of~$L$
  is the function mapping a word of $\Sigma^*$ to the element of $S$ that
  corresponds to its class in the syntactic equivalence relation.

  We will discuss the classes \LZG and \QLZG from~\cite{amarilliICALP21}. First, 
recall that \LZG is the
  variety of those semigroups for which all submonoids are in \ZG. Further,
  \QLZG are those languages whose so-called \emph{stable semigroup} is in \LZG.
  To define this, we call the \emph{stability index} of~$L$ the value $s \geq 1$
  which is the idempotent power of $\eta(\Sigma)$ in the powerset monoid
    of the syntactic monoid $M$ of~$L$. Then the \emph{stable semigroup} of~$L$
    is $\eta(\Sigma^s)$.
  We are now ready to explain how the proof of \cref{prp:zgupdates} follows from
  Theorem~1.1 in~\cite{amarilliICALP21}:

  \begin{proof}[Proof of \cref{prp:zgupdates}]
  We first show the upper bound, by explaining why a \ZG language $L$ is indeed in the
  class \QLZG. Indeed, the syntactic monoid $M$
  of~$L$ in in \ZG by hypothesis. Now, the stable semigroup $S'$ of~$L$ is
  $\eta(\Sigma^s)$, which is a subsemigroup of $M$. This implies that the
  submonoids of~$S'$ are submonoids of~$M$. As $M$ is in \ZG, and \ZG as a
  variety is closed under taking submonoids, we know that all submonoids of~$S'$
    are in \ZG, so the semigroup $S'$ is in \LZG, and $L$ is in \QLZG as stated.

  For the lower bound, we will show that for a regular language $L$ with a
  neutral letter~$e$, if $L$ is in \QLZG, then $L$ is in \ZG. The
  contrapositive of this result, together with the contrapositive of 
  the second point in our summary above of Theorem~1.1
  in~\cite{amarilliICALP21}, implies the lower bound that we have to establish.
  Note that the image $\eta(e)$ of the neutral letter~$e$ by the syntactic
    morphism~$\eta$ is an identity of the syntactic semigroup $S$. In
  particular this implies that $S$ is a monoid and is actually equal to the
  syntactic monoid $M$ of $L$.

  We first claim that the stable semigroup $S' \coloneq \eta(\Sigma^s)$ of~$L$ is in fact exactly the
  syntactic semigroup $S$ of~$L$. We show both inclusions. First, we have by
  definition that $S' \subseteq S$. Second, any element $x$ of $S$ can be
  achieved as $x = \eta(w)$ for some $w \in \Sigma^+$. As $L$ features a
  neutral letter~$e$, up to padding $w$ to the right with~$e$, we can assume
  that the length of~$w$ is a multiple of the stability index~$s$, and write $w = w_1 \cdots w_k$
  with $|w_i| = s$ for each $i$. Then
  $x = \eta(w) = \eta(w_1 \cdots w_k) = \eta(w_1) \cdots \eta(w_k)$ shows that
  $x \in \eta(\Sigma^s)$. Thus, we have established that $S'=S$.

  Now, by definition of $L \in \QLZG$ we know that its stable semigroup $S'$,
  i.e., its
  syntactic semigroup $S$, is in \LZG.  This means by definition that all local
  monoids of~$S$ are in \ZG. Consider the local monoid $fSf$ of~$S$, where
  $f=\eta(e)$ is the identity of~$S$. The definition implies that $S$ is in \ZG.
  Thus, we have shown that the syntactic monoid $M$ of~$L$, which is equal
  to~$S$, is in \ZG. This concludes the proof.
\end{proof}
\end{toappendix}

\section{Extensible Languages and Simple Infix Enumeration Algorithm}
\label{sec:extensible}

In this section, we introduce \emph{extensible} languages (also known as
\emph{two-sided ideals}~\cite{BrzozowskiY11}) as a simple class of
languages $L$ for which the dynamic enumeration problem of $L$-infixes can be reduced to the dynamic
membership problem in a simple way, in the presence of a neutral letter, and 
up to forgoing our constant memory enumeration requirement.
Let us first define
the notion of extensible languages:

\begin{definition}[Extensible language]
  \label{def:extensible}
  An \emph{extension} of a word $w \in \Sigma^*$ is any word $z \in \Sigma^*$
  such that $w$ is a factor of~$z$, i.e., we have $z = swt$ for some
  choice of $s,t\in\Sigma^*$. 
  A language~$L$ is \emph{extensible} if, for any words $w, z \in \Sigma^*$,
  if $w \in L$ and $z$ is an extension of~$w$ then $z \in L$.
\end{definition}

Equivalently, a language $L$ is extensible iff $L = \Sigma^* L \Sigma^*$.  The
key property of extensible languages is that, by contraposition, we have
$w \not\in L$ iff no infix of $w$ realizes a word of~$L$.  Hence testing
membership of $w$ to~$L$ is equivalent to testing whether the enumeration of
$L$-infixes on~$w$ produces results.

The extensibility of a language~$L$ can be used in two directions. First, by the above,
the dynamic membership problem for~$L$ immediately reduces to the dynamic
enumeration problem for $L$-infixes. We
will use this in \cref{sec:lower} to show lower bounds on the
latter problem. Second, if $L$ contains a neutral letter, we can do the converse
and reduce the dynamic enumeration of $L$-infixes to dynamic membership to~$L$. Namely:

\begin{propositionrep}
  \label{prp:simple}
  Let $L$ be an extensible language featuring a neutral letter. 
  Assume that
  we have a data structure $\Psi$
  for dynamic membership to $L$ 
  with linear preprocessing, linear memory, and constant-time updates and
  membership queries. 
  Then we can build a data structure for the dynamic enumeration problem for
  $L$-infixes with 
  linear preprocessing, linear memory, constant-time updates and constant
  delay enumeration.
\end{propositionrep}

In particular, in view of \cref{prp:zgupdates}, this implies the following:

\begin{corollary}
  \label{cor:zg}
  Let $L$ be an extensible \ZG language. Then the dynamic enumeration
  problem for $L$-infixes
  can be solved with linear preprocessing, linear memory, constant-time
  updates, and constant delay enumeration.
\end{corollary}

However, we do \emph{not} achieve constant memory enumeration.  During the
enumeration phase, the algorithm modifies the data structure by simulating
updates. It only restores it at the end of the enumeration.
Even worse, our algorithm cannot perform (real) updates
until the enumeration is over.
We sketch the proof of
\cref{prp:simple} below (see Appendix~\ref{apx:simple_algorithm}):

\begin{proofsketch}
  The data structure, denoted $\Psi$, is the one used to handle updates while
  maintaining dynamic membership to~$L$.  We start enumeration by querying
  $\Psi$ to determine if the current word $w$ belongs
  to~$L$. If not, $w$ has no $L$-infixes and we are done. Otherwise we can
  produce the infix $[1,|w|]$. Then we decrement the right endpoint and perform
  substitution updates (by the neutral letter) in~$\Psi$ to remove
  letters one by one and produce infixes of the form $[1,r]$. We stop when the current infix is
  no longer in~$L$ according to~$\Psi$. Then we add back the letters of~$w$
  to~$\Psi$
  and increment back~$r$ until reaching $|w|$.
  To ensure constant delay, we use the standard even-odd technique~\cite{uno2003two}:
  we produce infixes with even~$r$
  as we are decrementing $r$ and those with odd~$r$ as we are incrementing~$r$.
  We continue similarly for successive values $l \geq 1$ of the
  left endpoint, by increasing order; again the even-odd technique can be used
  to ensure that $\Psi$ is back to its initial state at the end of the
  enumeration.
\end{proofsketch}

In \cref{sec:main}, we provide a more
complex algorithm that fulfills our constant memory enumeration, and which
further applies to a larger class
of \emph{semi-extensible ZG languages}.

\begin{toappendix}
\label{apx:simple_algorithm}

\newcommand\update{\textsc{update}}
\newcommand\test{\textsc{test}}

\begin{algorithm}[p]
  \caption{Dynamic enumeration problem of $L$-infixes for an extensible language $L$ with
  neutral letter $e$, assuming an oracle $\Psi$ for dynamic membership to~$L$.}
\label{alg:simple}
\begin{algorithmic}[1]
\Require
  Word $w$, and dynamic membership data structure $\Psi$ (with
  $\Psi.\update$ and $\Psi.\test$ methods) initialized with~$w$.
\Ensure
  All the $L$-infixes of $w$
\For{$l \gets 1$ to $|w|$}
\Comment{increase left endpoint $l$ starting at 1}
    \If{$l$ is even} \Comment{produce infixes with even left endpoint when increasing $l$}
      \State $b \leftarrow$ \Call{EnumerateL}{$l$}  \Comment{see procedure
      \textsc{EnumerateL} below}
      \LComment{if left endpoint $l$ returned no infix at all, start moving backwards immediately}
      \If{\textbf{not}~$b$}
        \State \textbf{break}
      \EndIf
    \EndIf
    \State $\Psi.\update(e,l)$\Comment{Substitute $e$ at position $l$ in~$\Psi$}
\EndFor
\For{$l' \gets l$ to $1$}
\Comment{decrease left endpoint $l'$ starting at $l$}
    \State $\Psi.\update(w[l'],l')$\Comment{restore position $l'$ in~$\Psi$}
    \If{$l'$ is odd} \Comment{produce infixes with odd left endpoint when
    decreasing $l'$}
      \State \Call{EnumerateL}{$l'$} 
    \EndIf
\EndFor
  \LComment{We have called \textsc{EnumerateL}($l'$) for each $l'$ which may
  return results, and $\Psi$ is back to its previous state: we are done}
  \State \Return
  \Statex

\Procedure{EnumerateL}{$l$}
  \Statex ~~~~\,\textbf{Input:}
  Word $w$, position $l$, data structure $\Psi$ initialized with $e^{l-1} w[l,|w|]$
  \Statex ~~~~\,\textbf{Output:}
  Output all the $L$-infixes of $w$ with left endpoint $l$, and return a
  Boolean\\\null~~~~\,\phantom{\textbf{Output:}} to indicate if some $L$-infixes have been enumerated
    \State $r \gets |w|$
    \If{\textbf{not}~$\Psi.\test()$}
        \LComment{$w[l,|w|] \notin L$, so by
        extensibility we know that all infixes with smaller right endpoint are not in~$L$}
        \State \Return \textbf{false} 
    \EndIf
    \LComment{decrease right endpoint $r$, doing modifications in~$\Psi$,
      and output infixes with even~$r$}
    \While{$r \geq l$ and $\Psi.\test()$}
      \If{$r$ is even}
        \State \Call{Output}{$[l,r]$}
      \EndIf
      \State $\Psi.\update(e,r)$\Comment{substitute $e$ at position $r$}
        \State $r \gets r - 1$
    \EndWhile
    \LComment{increase right endpoint $r$,
      undoing the modifications in~$\Psi$, and output infixes with odd $r$}
    \While{$r \leq |w|$}
        \State $r \gets r + 1$
        \State $\Psi.\update(w[r],r)$\Comment{restore $w[r]$ at position~$r$}
        \If{$r$ is odd}
          \State \Call{Output}{$[l,r]$}
        \EndIf
    \EndWhile
    \State \Return \textbf{true}  \Comment{at least one infix produced}
\EndProcedure
\end{algorithmic}
\end{algorithm}

In the rest of this appendix, we prove \cref{prp:simple}. The complete algorithm
  is presented as \cref{alg:simple}.

Let $L$ be an extensible language featuring
a neutral letter $e$, and assume that we have a structure $\Psi$ for dynamic
  membership to~$L$ which stores the current word~$w$ and is maintained up to
  date as~$w$ is updated. To be precise, $\Psi$ provides a method $\Psi.\update(a,i)$ that updates the
data structure to substitute the letter $a$ at position $i$ in the word that it
  stores, and a method
$\Psi.\test()$ that checks whether the current word stored by~$\Psi$ is in $L$.

  We implement the requisite data structure for dynamic enumeration
  of $L$-infixes by performing the same preprocessing
  as that of~$\Psi$, and we support updates by reflecting them directly
  in~$\Psi$. All that remains is to show how to perform $L$-infix enumeration
  using the structure~$\Psi$.

When processing an enumeration query, we will perform updates on~$\Psi$ to reflect the
infixes that we are considering in the enumeration: specifically we will
erase some letters of~$w$ by substituting them by~$e$, and later substitute back
the original letters of~$w$.
Note that the underlying
word $w$ is not modified during the enumeration, so $\Psi$ intuitively goes ``out of sync'' with~$w$
during the enumeration: during that time, $\Psi$ reflects the status of other words
  of length $|w|$ which intuitively correspond to the infixes produced -- with
  letters outside the infix in question having been replaced by the neutral letter~$e$.

The intuition for the algorithm is now simple: we enumerate all infixes $[l,r]$
  sorted first by decreasing order
for $r$ and then by increasing order for $l$. When the infix $[l,r]$ is considered, the
current word is such that the neutral letter $e$ is present at all positions
$i<l$ or $i>r$. Hence testing whether the infix $[l,r]$ is in $L$ amounts to a
membership test for the entire word stored by~$\Psi$, using $\Psi.\test()$. If the test is positive
we output $[l,r]$, and continue with $[l,r-1]$ after placing $e$ at position
$r$ using $\Psi.\update(e,r)$.  If the test is negative we know that no infix of
$[l,r]$ can be in $L$ by extensibility, and we can start again with the infix
$[l+1, |w|]$ after placing $e$ at position $l$ using $\Psi.\update(e,l)$.

This simple idea does not work as-is, because when we go from $[l,r]$ to $[l+1,|w|]$
the current word has the letter $e$ at all positions $i>r$ where we would need
to have the initial letters back. As it would take superconstant time to restore
all these letters, we use the even-odd technique of~\cite{uno2003two} mentioned
in the proof sketch. Specifically, when $r$ is decreased we only output
the infixes with even $r$ until we reach the time when $[l,r]$ is not in
$L$. We then increase $r$, outputting the infixes with odd $r$ while
restoring the initial letters of the word.

Further, as the acute reader has probably noticed, the algorithm as described
  would terminate with the data structure $\Psi$ storing a word with letters $e$
  in its prefix, whereas we want to restore $\Psi$ to its correct state
  when we are done processing the enumeration query. This can be handled using for $l$
the same even-odd trick that has been used for $r$, restoring the initial word
step by step. Note that it is important that, when incrementing $l$, we stop as
  soon as we reach a value for which $w[l,|w|]$ is not in~$L$ (according
  to~$\Psi$), and start decrementing $l$ at that point. Specifically, we must
  not wait for $l$ to
  take value $|w|-1$ before we start to decrement, as the delay could then be linear in~$|w|$. We have now
  converged to the algorithm as presented in \cref{alg:simple}.

  To show correctness, it suffices to observe that the
enumeration algorithm outputs an infix every $B$ steps for some constant~$B$,
which uses the fact that the membership test of~$\Psi$ is in constant time. Hence the delay is, up to a
constant factor, the time taken by $\Psi.\test$ and $\Psi.\update$. What is more, the
algorithm is correct, because for every outputted infix the factor of~$w$
that it realizes is a word of~$L$ (according to~$\Psi$,
which stores a word which is identical to the outputted infix up to neutral
letters). Further every infix of~$w$ is considered to be potentially outputted (at a point which
depends on the parity of its left and right endpoints), unless we have already
witnessed that an extension of the factor realized by this infix is not in~$L$
according to~$\Psi$. This establishes the correctness of the algorithm and
concludes the proof of \cref{prp:simple}.

\end{toappendix}

\section{\THreshext Languages}
\label{sec:threshext}
\begin{toappendix}
\label{apx:threshext}
\end{toappendix}

We now define our 
class of \emph{\threshext languages}, which
generalize extensible \ZG languages.
We present three equivalent definitions:
one algebraic definition involving the \ZG equation, one
operational definition 
used in our algorithms, and one
self-contained definition. 
We then give our algorithm for dynamic enumeration of $L$-infixes in the next
section.

\subparagraph*{Algebraic definition.}
We give a first definition via the \ZG equation (recall that $\omega$ denotes
the idempotent power of $M$):

\begin{definition}[\Threshext language, algebraic def.]
  \label{threshdef-eq}
  A regular language $L$ is \emph{\threshext} if it is in ZG, i.e. its syntactic
  monoid $M$ satisfies the \ZG equation:
   \begin{align}
     \tag{ZG}
     \text{for all~} x,y \in M, \text{~we have~} y x^{\omega+1} &= x^{\omega+1} y\,,\label{eq:ZG}
   \end{align}
   and further $L$ is \emph{semi-extensible} in the following sense:
   \begin{align*}
     \text{for all words~} x, y, z \in \Sigma^* \text{~and all non-neutral letters~} a,
     \text{if~} y \in L \text{~then~} x y z a^\omega \in L\,.
   \end{align*}
   \end{definition}

   Extensibility implies semi-extensibility, so \threshext
   languages generalize extensible \ZG languages which we studied in the
   previous section. The inclusion is strict: the language $L$ on alphabet
   $\Sigma = \{a,e\}$ that has neutral letter~$e$ and consists of the words
   with no $a$'s or at least two $a$'s is \threshext but not extensible.  It is
   also easy to deduce from the definition that
\threshext languages are always aperiodic, i.e., their syntactic monoid
satisfies the equation $x^\omega = x^{\omega+1}$.
We will give an example of a semi-extensible \ZG language later in this section
once we have stated our three equivalent definitions of the class.

\subparagraph*{Operational definition.}
We give another definition of \threshext languages for our
enumeration algorithms.
In the definitions below, for any subalphabet $S \subseteq \Sigma$ and word
$u \in \Sigma^*$, we denote by~$u_{-S}$ the word over
$\Sigma \setminus S$ obtained by removing all letters in~$S$ from~$u$. (In
particular if $S = \Sigma$ then we have $u_{-S} = \epsilon$ for any~$u$.)

\begin{definition}[\Threshext language, operational def.]\label{threshdef-condp}
  Let $L$ be a regular language. For each non-empty set $S$ of non-neutral
  letters for $L$, let $\cond(S) \subseteq \Sigma^*$ be the (regular) language
  of all words $u$ for which there is a word $v\in L$ such  that $v_{-S}$ is a factor of~$u_{-S}$.

  Then the language $L$ is \threshext if there is a number $p \in \NN$, called
  the \emph{threshold}, satisfying the following.
  We require that for each word $u \in \Sigma^*$,
  letting $T$ be the set of non-neutral letters $a$ of~$\Sigma$ such that
  $|u|_a \geq p$ ,
  we have either $T = \emptyset$ or we have that:
  \[u_{-T} \in \cond(T) \text{ iff } u\in L\]
\end{definition}

In the definition above, notice that 
for any non-empty set $S$ of non-neutral letters, we have that
all letters of $S$ are neutral for
$\cond(S)$,
and further
if $u_{-S}\notin \cond(S)$ then $u\not\in L$.
Indeed, one can show the
contrapositive by taking $v=u$ in the definition of $\cond(S)$.

\subparagraph*{Self-contained definition.}
We last give a self-contained definition:

\begin{definition}[\Threshext language, self-contained def.]\label{threshdef-mon} A regular language is \threshext if there exists a number $p \in \NN$
  such 
  that the following holds for all words $u, v \in \Sigma^*$: letting $T$ be
  the set of non-neutral letters $a$ of~$\Sigma$ such that $|u|_a \geq p$, if $v \in L$ and
  $T$ is non-empty and $v_{-T}$ is a factor of $u_{-T}$ then $u\in L$.
\end{definition}

We give an example illustrating the power of \threshext languages: 
\begin{example}
  \label{exa:threshext}
  Consider the language $L_{ab}$ on alphabet $\Sigma = \{a, b, e\}$
  that has $e$ as a neutral letter and contains the words $u \in
  \Sigma^*$ satisfying one of the following: (1.) $u_{-\{e\}} = ab$
  (i.e., $u$ is exactly the word $ab$ up to the neutral letters); or (2.) $|u|_a \geq 3$
  ($u$ has a least 3 $a$'s); or (3.) $|u|_b \geq 3$ ($u$ has at least 3 $b$'s).
  The language $L_{ab}$ is not extensible (e.g., $ab \in L_{ab}$ but $abb \notin
  L_{ab}$), but it is \threshext with threshold $p=3$. 
    Indeed, by the characterization of \cref{threshdef-mon},
  let $v, u \in \Sigma^*$ be words such that
    $v \in L_{ab}$, such that $T = \{a \neq e \mid |u|_a \geq p\}$ is non-empty,
    and such that $v_{-T}$ is a factor of $u_{-T}$. We have $u \in L_{ab}$ because
    $u$ contains at least $3$ occurrences of a non-neutral letter (i.e., from
    the nonempty set~$T$).

  The languages $\cond$ in \cref{threshdef-condp} for
  $L_{ab}$ are simply
  $\cond(\{a,b\})=\cond(\{a\})=\cond(\{b\})=\Sigma^*$
  because all words are in $L_{ab}$ when some non-neutral letter occurs
  $\geq 3$ times.
\end{example}

\subparagraph*{Equivalence of definitions.}
We can show in Appendix~\ref{apx:threshext}
that the three definitions are equivalent:

\begin{propositionrep}\label{thm-threshold-monotone-def-equiv}
  \cref{threshdef-mon,,threshdef-condp,,threshdef-eq} define the same family of
  languages.
\end{propositionrep}

\begin{proof}
  We show the following implications: \cref{threshdef-mon} $\Rightarrow$
  \cref{threshdef-condp} $\Rightarrow$
  \cref{threshdef-eq} $\Rightarrow$ \cref{threshdef-mon}. For an integer $p \in
  \NN$ and a word~$w$, we partition
  $\Sigma$ into neutral letters, rare letters, and frequent letters, as was
  explained in the main text above \cref{lem:condefficient}.

  \medskip

  $\bullet$ {\bf \cref{threshdef-mon} $\Rightarrow$ \cref{threshdef-condp}.}

  Assume $L$ satisfies \cref{threshdef-mon}, and let $p \in \NN$ be a
  threshold.
  We show that the condition of \cref{threshdef-condp} is satisfied for this value of $p$.

  Let $u$ be a word and let $T \subseteq \Sigma$ be the set of frequent
  letters in~$u$, i.e., for each non-neutral $a \in \Sigma$ 
  we set $a \in T$ iff $|u|_a \geq p$. Assume $T\neq \emptyset$ as there is
  nothing to show otherwise.

  If $u_{-T} \not\in \cond(T)$, then we conclude that $u \notin L$ by
  construction of $\cond(T)$, as explained just after \cref{threshdef-condp}.

  If $u_{-T} \in \cond(T)$ then by definition of $\cond(T)$ there exists a word $v\in L$ such
  that $v_{-T}$ is a factor of $u_{-T}$. We get $u\in L$ by \cref{threshdef-mon}.

  \medskip

  $\bullet$ {\bf \cref{threshdef-condp} $\Rightarrow$ \cref{threshdef-eq}.}

  Assume $L$ satisfies \cref{threshdef-condp}, and let $p \in
  \NN$ be a threshold. We need to show that the
  conditions of \cref{threshdef-eq} holds.
  Let $M$ be the syntactic monoid of~$L$,
  let $\alpha$ be the syntactic morphism
  of $L$,
  and let $\omega$ be the idempotent
  power of~$M$.
  Without loss of
  generality, up to replacing $\omega$ by a multiple which is also an
  idempotent power, we can assume that $\omega>p$.
  
  To show \cref{eq:ZG} of \cref{threshdef-eq}, let
  $X,Y \in \Sigma^*$ be words such that $\alpha(X)=x$ and $\alpha(Y)=y$, let
  $u,v\in\Sigma^*$ be arbitrary words, and 
  consider $w=uYX^{\omega+1} v$ and $w'=uX^{\omega +1}Yv$. We must show that $w
  \in L$ iff $w' \in L$.
  Let $T$ be the set of 
  frequent letters of $w$ for threshold~$p$. Notice
  that 
  $T$ is also the set of frequent non-neutral letters of $w'$.

  Let $\hat w$ and $\hat w'$ be the words constructed from $w$ and $w'$ by
  removing the neutral letters, in particular $w\in L$ iff $\hat w\in L$ and
  $w'\in L$ iff $\hat w'\in L$. Notice also that by construction
  $\hat w_{-T}=\hat w'_{-T}$, because all letters in the factor $X^{\omega+1}$
  occur at least $p$ times, so either they are neutral or they are in~$T$.
  Further, if $T = \emptyset$ then it means $X$ consists only of neutral
  letters, and this implies $\hat w=\hat w'$ in which case we immediately have
  $w \in L$ iff $w' \in L$. Hence in the sequel we assume that $T \neq
  \emptyset$.
  
  By \cref{threshdef-condp} we also have $\hat w_{-T} \in \cond(T)$ iff $\hat
  w\in L$ and $\hat w'_{-T} \in \cond(T)$ iff $\hat w'\in L$.
  Hence we have $w \in L$ iff $\hat w \in L$ iff $\hat w_{-T} \in \cond(T)$ iff
  $\hat w'_{-T} \in \cond(T)$ iff $\hat w' \in L$ iff $w' \in L$.  This chain of
  equivalences establishes that \cref{eq:ZG} holds.
  
  To show the semi-extensibility property on $L$, consider $w\in L$.
  Let $u,v \in \Sigma^*$ be words and let $a$ be a non-neutral letter. We need to show that $uwva^\omega \in
  L$. Let $T$ be the set of frequent letters of $uwva^\omega$, and note that $T$ is not empty
  as it contains $a$. Further, the word $w_{-T}$ is a factor of
  $(uwva^\omega)_{-T}$ and we have $w\in L$.
  This implies that $(uwva^\omega)_{-T} \in \cond(T)$, so $uwva^\omega \in L$ 
  by \cref{threshdef-condp}
  as desired.

  \medskip
  
  $\bullet$ {\bf \cref{threshdef-eq} $\Rightarrow$ \cref{threshdef-mon}.}

  Assume that $L$ satisfies the conditions of \cref{threshdef-eq}. Take
  $p$ to be greater than the cardinality of the syntactic monoid of~$L$.
  Let $w$ be a word and let $T$ be its set of frequent non-neutral
  letters. Assume $T$ is not empty, as otherwise there is nothing to show.
  Let $u\in L$ and assume that $u_{-T}$ is  a factor of $w_{-T}$. We need to show
  that $w\in L$.

  Let $a\in T$.
  As $|w|_a \geq p$, we can write $w=w_0aw_1aw_2\cdots a
  w_k$, for $k>p$. As $k$ is greater than $p$ which is greater than the
  cardinality of the syntactic monoid, a classical pumping argument
  using the pigeonhole principle shows that $w
  \in L$ iff $w'\in L$ where $w'=w_0a\cdots aw_{i-1}(aw_i\cdots aw_j)^\omega
  aw_{j+1}\cdots $.

  Let us now observe that it follows from \Cref{eq:ZG} that $(xy)^\omega=x^\omega
  (xy)^\omega$. Indeed, it is known 
  by \cite[Lemma 3.8]{amarilli_paperman_ZG} that the following (*) is true:  $(xy)^\omega = x^\omega
  y^\omega$. Thus, we have $(xy)^\omega = (xy)^\omega(xy)^\omega =
  x^\omega y^\omega (xy)^\omega$ by (*), and this is equal to $x^\omega x^\omega
  y^\omega (xy)^\omega$ which by (*) again is equal to $x^\omega (xy)^\omega
  (xy)^\omega = x^\omega (xy)^\omega$, establishing the desired equation 
  $(xy)^\omega=x^\omega (xy)^\omega$.

  From this and \Cref{eq:ZG},
  letting $T_1=T\setminus\{a\}$,
  it follows that $w'\in L$ iff $w''\in L$, where
  $w''= w_{1} a^{\omega}$ and where $w_1$ is constructed from $w_{-T}$ by adding
  an occurrence of $a$ wherever it is needed for $u_{-T_1}$ to be a factor of
  $(w_1)_{-T_1}$. The reason is that, intuitively, the equation from the
  previous paragraph allows us to add a factor $a^\omega$, which by \Cref{eq:ZG}
  is central. Using the aperiodicity of the syntactic monoid of~$L$ (which
  follows from \Cref{eq:ZG} and semi-extensibility, as explained after
  \cref{threshdef-eq}), the $a^\omega$ allows us to collect all occurrences
  of~$a$ and insert $a$'s wherever needed, without affecting membership to~$L$.
  So indeed we can lift our hypothesis ``$u_{-T}$ is a factor of~$w_{-T}$'' to
  ``$u_{-T_1}$ is a factor of~$(w_1)_{-T_1}$'', 

  Repeating this argument for all the remaining letters of $T$ eventually
  yields a word $v$ such that $w\in L$ iff $v \Pi_{a \in T} a^\omega \in L$ and
  where $\hat u$ is a factor of $\hat v$, for $\hat u$ and $\hat v$ 
  respectively obtained from~$u$ and~$v$ by removing the neutral letters.
  
 From this and from the semi-extensibility of $L$ (given that $T \neq
 \emptyset$) we get that $v\Pi_{a \in T} a^\omega  \in
 L$, hence $w \in L$, which concludes.
 \end{proof}

\subparagraph*{Properties of \threshext languages.}
We close the section with two immediate properties 
that we will use
in our algorithms in the
next section. We focus on the operational definition,
\cref{threshdef-condp}, which we use in the sequel.

Based on \cref{threshdef-condp}, let $L$ be a \threshext language, 
and pick a suitable threshold $p \in \NN$. Up to increasing $p$, we assume without loss of
generality that $p\geq 2$. 
When considering a word $w
\in \Sigma^*$, we will partition the alphabet $\Sigma$ into three kinds of
letters: the \emph{neutral} letters (which have no impact for membership
to~$L$), the
\emph{rare} letters $a \in \Sigma$ such that $|w|_a < p$, and the
\emph{frequent} letters $a \in \Sigma$ such that $|w|_a \geq p$.

We first observe that membership to the languages $\cond(S)$ in \cref{threshdef-condp} can be
verified in constant time on words with a small number of non-neutral letters:

\begin{lemmarep}
  \label{lem:condefficient}
  Consider a \threshext language $L$ with threshold $p$. We can
  precompute a data structure allowing us to do the following in constant time:
  given a non-empty set $S \subseteq \Sigma$ of non-neutral letters, given a
  word $w$ on alphabet $\Sigma \setminus S$ with at most $p$ occurrences of each letter, determine in constant time whether $w \in \cond(S)$.
\end{lemmarep}

\begin{proofsketch}
  We simply tabulate explicitly the set $\cond(S)$ on words with at most $p$
  occurrences of each non-neutral letter, as their size is constant.
\end{proofsketch}

\begin{proof}
  Remember that we measure complexity only as a function of the length of the
  input word, with the language $L$ being fixed. In particular, the threshold
  $p$ is also fixed. So the result is immediate as we can prepare a table of constant size which tabulates
  the answer to all possible queries: for each all non-empty subsets $S \subseteq
  \Sigma$ of non-neutral letters, for each word $w$ on $\Sigma \setminus S$
  with at most $p$ occurrences of each letter,
  we simply store the information of whether $w \in \cond(S)$.
\end{proof}

Second, we show that membership to $\cond(T)$, for $T$ the set of frequent
letters of a word~$w$, is a necessary condition for $w$ to contain factors
belonging to~$L$. This allows us to stop enumeration early whenever the
current infix realizes a factor $w \notin \cond(T)$:

\begin{lemmarep}
  \label{lem:condmonotone}
  Consider a \threshext language $L$.
  Let $w$ be a word and $T$ be the set of frequent letters of $w$.
  If $T \neq \emptyset$ and $w \not\in \cond(T)$, then no factor
  of~$w$ belongs to~$L$.
\end{lemmarep}

\begin{proof}
  Assume by contradiction that $w$ has a factor $y$ such that $y$ belongs
  to~$L$. In particular, $y_{-T}$ is a factor of $w_{-T}$.
  Now, by definition of $\cond(T)$,
  taking $v\coloneq y$ in the definition, we have $w \in \cond(T)$, which
  contradicts our assumption.
\end{proof}

\section{Infix Enumeration Algorithm for \THreshext Languages}
\label{sec:main}

Having defined \threshext languages, we show the following in this section.
This is the main technical result of the paper, and generalizes \cref{cor:zg}:

\begin{theorem}
  \label{thm:main}
  Let $L$ be a \threshext language. Then the dynamic enumeration problem of
  $L$-infixes 
  can be solved 
  with constant-time updates, constant delay enumeration, constant-memory
  enumeration, linear preprocessing, and linear memory.
\end{theorem}

In the rest of this section, we prove \cref{thm:main}.
We let $w$ be the word on which the algorithm runs, and $n \coloneq |w|$ be its
length, which never changes.
We first present the
preprocessing and the support for update operations in \cref{sec:occlists}, 
with the classical data structure of \emph{occurrence lists}. The main technical
challenge is to support 
$L$-infix enumeration queries with constant delay and constant 
memory, which we describe in \cref{sec:algo}.

\subsection{Preprocessing and Updates with Occurrence Lists}
\label{sec:occlists}
\begin{toappendix}
  \subsection{Preprocessing and Updates with Occurrence Lists}
  \label{apx:occlists}
\end{toappendix}

Occurrence lists are a standard data structure which store
elements in an array of size~$n$, coupled with a doubly-linked list to
efficiently iterate over the elements in insertion order (similar, e.g., to the
LinkedHashSet data structure in Java). Formally:

\begin{definition}
  \label{def:occurrencelist}
  An \emph{occurrence list} is a data structure $\Lambda$ that stores a set $S$ of integers of
  $\{1, \ldots, n\}$ and supports the following operations:
  \begin{itemize}
    \item Insert: given $i \in \{1, \ldots, n\} \setminus S$,
      add $i$ to~$S$;
    \item Delete: given $i \in S$, remove $i$ from~$S$;
    \item Count: get the size $|S|$ of~$S$;
    \item Retrieve: enumerate all the elements currently in~$S$.
  \end{itemize}
\end{definition}

The following result is well known, but for self-containedness we provide a proof
in Appendix~\ref{apx:occlists}.

\begin{propositionrep}[Folklore]
  \label{prp:occlist}
  We can implement occurrence lists to take memory $O(n)$ and to ensure that all operations take $O(1)$
  computation steps, except the Retrieve operation which enumerates the
  elements with $O(1)$ delay.
\end{propositionrep}

We clarify that the data structure stores elements 
in insertion order and not in increasing order, so in particular 
it cannot be used as a
constant-time priority queue.

\begin{proof}
An occurrence list can be implemented by a doubly-linked list $\Lambda'$ together with an
array $T$ of size~$n$.
The list $\Lambda'$ contains the elements of~$\Lambda$ in insertion order,
  and for each element $i \in
\{1, \ldots, n\}$ the $i$-th cell $T[i]$ of the array $T$ contains a pointer to the list
item of~$\Lambda'$ that represents the element~$i$. The other cells are empty. We additionally maintain the count of elements
separately for the Count operation.

  Then, we can implement Insert by adding the
provided element $i$ as a new list item in $\Lambda'$ and making $T[i]$ point to that
item, we can implement Delete by looking up $T[i]$ for the provided element $i$
to remove its list item from $\Lambda'$ (in constant time because $\Lambda'$ is a
doubly-linked list) and clearing $T[i]$, and we can perform Retrieve by
iterating over $\Lambda'$.
\end{proof}

\subparagraph*{Preprocessing and updates.}
We can now explain how occurrence lists are used for the algorithm that we
use to show \cref{thm:main}.
The algorithm will use an occurrence list $\Lambda_a$ for each
letter $a \in \Sigma$, to store
the set of positions of $w$ that contain the letter
$a$. We denote by $\Lambda$ the tuple consisting of all occurrence lists
$\Lambda_a$ for $a\in \Sigma$.

The preprocessing of the algorithm simply consists of building $\Lambda$ on the
input word~$w$, in linear time and with linear memory by \cref{prp:occlist}.
The handling of letter substitution updates simply amounts to maintaining
$\Lambda$ up to date. Formally, 
let $(i,a)$ be a letter substitution update, and let
$b \in \Sigma$ be the letter $w[i]$ before the update (which can be obtained in
constant time and linear memory up to storing the current $w$ in a table).
We handle the update by applying it to $\Lambda$,
namely, we remove $i$ from $\Lambda_b$ and add $i$ to $\Lambda_a$, in
constant time by \cref{prp:occlist}.

\tikzset{
    structure/.style={
        rectangle,
        fill=T-Q-PH1,
        append after command={
            \pgfextra{
                \draw[T-Q-PH1,line width=1pt] ([xshift = .5cm]\tikzlastnode.north east) -- ([xshift = .5cm]\tikzlastnode.south east);
            }
        }
    }
}

\tikzset{
    counting/.style={
        rectangle,
        fill=T-Q-PH3,
        append after command={
            \pgfextra{
                \draw[T-Q-PH3,double,line width=1pt] ([xshift = 1cm]\tikzlastnode.north east) -- ([xshift = 1cm]\tikzlastnode.south east);
            }
        }
    }
}

\begin{algorithm}
\caption{Enumerate $L$-infixes of a word.}
  \label{alg:antoineenum}
\begin{algorithmic}[1]
\Require \threshext language $L$, threshold $p \geq 2$, languages $\cond(S)$ 
   for each set of non-neutral letters $S$ (see \cref{threshdef-condp}), word $w$, occurrence lists $\Lambda_a$ for $a \in
  \Sigma$, 
\Ensure enumeration of the $L$-infixes of $w$

\BeginBox[counting]\State $\text{nocc}, \text{nocc}' \gets$ empty arrays \tabto{\linewidth}\null\EndBox
\ForAll{$a \in \Sigma$}
  \BeginBox[counting]\State $\text{nocc}[a] \gets |\Lambda_a|$
  \Comment{number of occurrences from occurrence lists}
  \label{lin:noccinit}\EndBox
  \State \Call{StartMinBackgroundTraversal}{$a$, $\Lambda_a$}
  \label{lin:minback}
\EndFor

\State $\mathfrak{L}_{-1}, \mathfrak{R}_{-1} \gets$ Unset

\BeginBox[structure]\For{$l \gets 1$ to $|w|$}
  \State $r \gets |w|$ \tabto{\linewidth}\null\label{lin:initr}\EndBox
  \BeginBox[counting]\State $\text{nocc}' \gets$ copy of $\text{nocc}$ \label{lin:nocccopy} \tabto{\linewidth}\null\EndBox
  \State $\mathfrak{R} \gets$ \Call{InitRI}{\,} \Comment{initialize current RI,
  see \cref{prp:ri}}
  \label{lin:initrs}
  \LComment{initialize $\Delta$ for all rare letters}
  \State $\Delta \gets$ \Call{GetRareOccurrences}{$l$, $r$, $\mathfrak{L}_{-1}$, $\mathfrak{R}_{-1}$, $\text{nocc}'$} \label{lin:getrareocc} %
  \BeginBox[counting]\State $T \gets \{a \in \Sigma \mid a \text{~non-neutral and~} \text{nocc}'[a] \geq p\}$\tabto{\linewidth}\null\EndBox
  \BeginBox[structure]\While{$T \neq \emptyset$ \textbf{and} word stored in
  $\Delta$ is in $\cond(T)$}
  \Comment{relies on \cref{lem:condefficient}}
  \label{lin:whilecond}
    \State \Call{Output}{$[l, r]$}\Comment{as $w[l,r]\in L$} \label{lin:output}\EndBox
    \BeginBox[counting]\State $a \gets w[r]$\Comment{this occurrence of $a$ will exit the right endpoint}
    \State $\text{nocc}'[a] \gets \text{nocc}'[a] - 1$
  \label{lin:noccrdec}
    \State $ T \gets \{a \in \Sigma\mid a \text{~non neutral and~} \text{nocc}'[a] \geq p\}$\EndBox
    \If{$\text{nocc}'[a] == p-1$ and $a$ non-neutral}
      \LComment{letter $a$ just became rare; populate $\Delta$ for $a$ (see
      \cref{alg:getRareOccurrences})} %
      \State $\Delta[a] \gets$ \Call{GetRareOccurrencesSingle}{$a$, $l$, $r-1$, $\mathfrak{L}_{-1}$, $\mathfrak{R}_{-1}$}
  \label{lin:getrareoccs}
    \EndIf
    \If{$\text{nocc}'[a] < p-1$ and $a$ non-neutral}
      \State $\Delta[a] \gets \Delta[a] \setminus \{r\}$
      \Comment{remove the occurrence of $a$ from occurrence list $\Delta[a]$}
      \label{lin:updatelists}
    \EndIf
    \State $\mathfrak{R} \gets$ \Call{UpdateRI}{$\mathfrak{R}$, $r$}
    \Comment{see \cref{prp:ri}}
    \label{lin:updaters}
\BeginBox[structure] \State $r \gets r - 1$ \label{lin:decrement}
    \Comment{move right endpoint $r$ one step left}
\EndBox
  \EndWhile
\BeginBox[structure]
  \If{$T \neq \emptyset$}
  \Comment{\nocond: by \cref{lem:condmonotone} no more $L$-infixes for left
  endpoint~$l$}
  \label{lin:nocond}
  \If{$r == |w|$}
    \LComment{while-loop exited immediately, so we must finish directly}
    \State \Return \Comment{exit as there are no more $L$-infixes at all by \cref{lem:condmonotone}}
    \label{lin:nocond2}
  \EndIf
  \Else
    \Comment{\allrare: there may be more infixes to produce with left
    endpoint~$l$}
    \If{$r == |w|$} \label{lin:allrare}
  \label{lin:finish}
      \LComment{while-loop exited immediately, so we must finish directly}
      \label{lin:allrare2}
      \State \Call{ProduceRemaining}{$l, |w|, \Delta$}
      \Comment{see \cref{lem:easycase} in Appendix~\ref{section-easy-enum}}
      \State \Return \Comment{all remaining infixes produced}
    \Else
      \LComment{there may be more infixes $[l,r']$ to produce with $r'\leq r$}
      \label{lin:allrare1}
      \State \Call{ProduceRemainingL}{$l,r,\Delta$}
      \Comment{see \cref{lem:easycase} in Appendix~\ref{section-easy-enum}}
    \EndIf
  \EndIf
  \label{lin:finishend}
  \EndBox
  \If{$l == 1$}
    \State $r_1 \gets r+1$ \Comment{store right limit $r_1$ for $l=1$}
    \label{lin:r1}
    \ForAll{$a \in \Sigma$}
      \State \Call{StartMaxLeftBackgroundTraversal}{$a$, $\Lambda_a$, $r_1$}
\label{lin:maxback}
    \EndFor
    \State $\mathfrak{L} \gets$ \Call{InitLI}{$r_1$} \Comment{initialize LI}
    \label{lin:initls}
  \Else
    \State $\mathfrak{L} \gets$ \Call{UpdateLI}{$\mathfrak{L}_{-1}$, $\mathfrak{R}_{-1}$}
    \Comment{compute LI at current $r_l$, see
    \cref{clm:maintainls}}
    \label{lin:updatels}
  \EndIf

  \State $\mathfrak{R}_{-1} \gets \mathfrak{R}$; $\mathfrak{L}_{-1} \gets
  \mathfrak{L}$ \Comment{current LI/RI become previous LI/RI}
  \BeginBox[counting]\State $\text{nocc}[w[l]] \gets \text{nocc}[w[l]] - 1$ \Comment{left endpoint $l$ will move
  right and one letter exits} \label{lin:noccldec}\EndBox
\EndFor
\end{algorithmic}
\end{algorithm}

\subsection{Main Enumeration Algorithm}
\label{sec:algo}

We now present how our algorithm handles $L$-infix
enumeration queries. The overall enumeration pseudocode is given in
\cref{alg:antoineenum} (with some helper functions explained
later).

\subparagraph*{High-level algorithm (blue).}  The core of the enumeration
algorithm 
is highlighted in blue (also with a single blue line
in the right margin) in \cref{alg:antoineenum}.  We consider all possible
infixes $[l,r]$ of $w$: by increasing order for $l$, from $1$ to $|w|$ (outer
for-loop); and, for each $l$, by decreasing order for $r$, from $|w|$ to $l$
(inner while-loop). When doing so, 
we maintain the list $T$ of all frequent
letters within $w[l,r]$ and an occurrence list $\Delta$ of all rare letters
within $w[l,r]$ (this is the main technical difficulty that will be discussed
below).  The while-loop runs until $T$ is empty or
$w[l,r]_{-T} \notin \cond(T)$ (the condition on line~\ref{lin:whilecond}, that can
be tested in constant time by \cref{lem:condefficient}). Until then, by
\cref{threshdef-condp}, 
$[l,r]$ is an $L$-infix and we output it.

When exiting the while-loop, two cases may arise.
The first case, called \nocond, is when $T \neq \emptyset$ but
$w[l,r]_{-T} \notin \cond(T)$. At this point, by~\cref{lem:condmonotone}, we
know that no smaller $r$ can yield an $L$-infix for this value of~$l$, and we
can safely move on to $l+1$.
(In the case where $r=|w|$, by \cref{lem:condmonotone} we can stop the whole
process, cf line~\ref{lin:nocond2}.)

The second case, called \allrare, is when $T= \emptyset$
(line~\ref{lin:allrare}), i.e., all non-neutral letters are now rare in
$w[l,r]$. The key point is that in this case $\Delta$ contains all the
non-neutral letters of $w[l,r]$ so enumerating the desired infixes is rather
simple. %
For technical
reason we consider two subcases.  When $r == |w|$
(cf. line~\ref{lin:allrare2}), we can use $\Delta$ to complete the whole
enumeration.
When $r<|w|$,
(cf. line~\ref{lin:allrare1}), we use $\Delta$ to produce all infixes of the
form $[l,r']$ with $r'\leq r$ and we then move on to~$l+1$.

There are two remaining challenges for the algorithm.
The first one is described
in the next paragraph, and the second is sketched in the rest of this section.

\subparagraph*{Challenge 1: Maintaining occurrence counts (green).}  The first
challenge is that we must keep track, for each considered infix $[l,r]$, of the
set $T$ of frequent letters (i.e., non-neutral letters that occur $\geq p$
times). Achieving this is completely routine, and is highlighted in green (also
with a double green line in the right margin) in \cref{alg:antoineenum}.  At
the beginning, when $w[l,r]=w$, we initialize $T$ using $\Lambda$
(cf. line~\ref{lin:noccinit}) which has been computed and maintained in
\cref{sec:occlists}.  Then we create a copy for each new value of $l$ (cf. 
line~\ref{lin:nocccopy}), and decrement the count for the letter $w[r]$ when $r$
decreases (cf. line~\ref{lin:noccrdec}), and decrement the counts for the
letter $w[l]$ when $l$ increases (cf. line~\ref{lin:noccldec}).

This way, we maintain the set $T$ of the frequent letters in the
current infix; this is used in particular to know to which language $\cond(T)$
the condition of the while-loop applies.

\subparagraph*{Challenge 2: Keeping track of the rare letter occurrences.}
The second challenge, which is far harder, is that we must know the positions of the
rare letters in the current infix $w[l,r]$ (and store them in occurrence
lists $\Delta$). This is used to test membership to
$\cond(T)$, and it is used when enumerating remaining infixes
in \allrare.

The lists $\Delta$ are small, because each rare letter occurs
$< p$ times, so the memory usage of $\Delta$ is constant overall.  However,
finding the positions themselves is our main technical challenge, which is
addressed by the non-highlighted code in \cref{alg:antoineenum}.
We first explain the case of the first iteration of the for-loop ($l = 1$), and then explain the case $l > 1$.

\begin{toappendix}

\subsection{Enumerating $L$-Infixes When There Are Only Rare Letters}\label{section-easy-enum}

We take care in this appendix of the case when $w[l,r]$ contains only rare
letters (\allrare in \cref{alg:antoineenum}).
In this case we have occurrence lists $\Delta$ that store the occurrences of these rare letters in $w[l,r]$.
Then we can easily enumerate all $L$-infixes of~$w[l,r]$. Formally:

\begin{lemmarep}
  \label{lem:easycase}
  Let $L$ be a language, and let
  $S \subseteq \Sigma$ be the set of non-neutral letters.
  For any fixed $c \in \NN$, the following is true.
  Given as input a word $w$ over $\Sigma$, an infix $[l,r]$, and
  occurrence lists $\Delta = \{\Delta_a \mid a \in S\}$ for $w[l,r]$ 
  with $|\Delta_a| \leq c$ for each $a \in S$, we can do the following
  with constant delay and constant memory:
      \begin{itemize}
        \item \textsc{ProduceRemaining}($l,r,\Delta$): enumerate all $L$-infixes $[i,j]$ of~$w$
          with $l \leq i \leq j \leq r$.
        \item \textsc{ProduceRemainingL}($l,r,\Delta$): enumerate all $L$-infixes $[l,j]$ of~$w$ with $l \leq j \leq r$.
      \end{itemize}
\end{lemmarep}

\begin{proofsketch}
  We use the occurrence lists to build the subword~$w'$ of the non-neutral
  letters in~$w[l,r]$: $w'$ is of constant size. Then we consider
  each infix of~$w'$ (there are constantly many) and test their membership
  to~$L$ (in constant time as they have constant length). For each infix
  of~$w'$ in~$L$, we easily produce the matching infixes of~$w$,
  by considering all possible ranges of neutral letters that can be added
  around non-neutral letters.
\end{proofsketch}

\begin{proof}
  We show the claim for the first point; the second point follows 
  by a
  similar argument.

  Let $N$ be the set of neutral letters
  of~$\Sigma$. From the occurrence lists $\Delta$,
  we prepare the subword $w' = w_{-N}$ of the non-neutral letters of~$w[l,r]$: by
  assumption it has constant size, i.e., the length $n'$ of~$w'$ is at most
  $ck$, so this computation takes constant time and memory. Further, the
  occurrence lists also give us the increasing sequence of indices $p_1, \ldots,
  p_{n'}$ corresponding to the occurrences of the letters of $\Sigma \setminus N$
  in $w$, in-order. We set $p_0 \coloneq l$ and $p_{n'+1} \coloneq r+1$ for convenience.

  We consider every infix $[i,j]$ of $w'$ with $1 \leq i \leq j \leq |w'|$, of
  which there are constantly many. Each of them has constant length, so we
  can check in constant time whether $w'[i,j]$ belongs to~$L$ or not. If it does,
  then we easily enumerate all infixes $[l',r']$ with $p_{i-1} + 1 \leq l' \leq
  p_i$ and $p_j \leq r' \leq p_{j+1}-1$. Note that these infixes are precisely
  the infixes of~$w[l,r]$ whose non-neutral letters are the letter occurrences of
  $w'[i,j]$, so each of these is indeed a result because $w'[i,j] \in L$ and by
  definition of a letter being neutral.

  Last, if $L$ contains the empty word, for each $1 \leq i \leq n'+1$, if $p_i >
  p_{i-1} + 1$, we 
  enumerate all infixes of the form $[l',r']$ with
  $p_{i-1} + 1 \leq l' \leq r' \leq p_i-1$ where all letters are neutral.

  The correctness of the algorithm is easy to establish. In one
  direction, every enumerated infix indeed belongs to~$L$. In the other
  direction, every infix
  of~$w[l,r]$ was considered: either it does not include one of 
  the letter occurrences of~$w'$ (i.e., it consists only of neutral letters, and
  was considered in the last paragraph), or it includes a non-empty infix of~$w'$ (and
  was considered in the preceding paragraph). Further, the constant
  delay and constant memory requirements are clearly satisfied. This establishes
  correctness and concludes the proof.
\end{proof}

\end{toappendix}

\subparagraph*{Finding rare letters for \bm{$l=1$}.}
We explain how to obtain the occurrence list
$\Delta_a$ of a non-neutral letter~$a$ that becomes rare in~$w[l,r]$ in the first iteration of the for-loop, i.e.,
$l=1$.
The letter $a$ may be \emph{initially rare}, i.e., already when $r == |w|$
(line~\ref{lin:getrareocc}), or it may become rare when $r$ is decremented
(line~\ref{lin:getrareoccs}). The case of an initially rare letter will in fact
be subsumed by the case of a letter that becomes rare, so we only discuss the
latter.

When $a$ becomes rare (line~\ref{lin:getrareoccs}) for some value of~$r$, we need to find $p-1$
occurrences of~$a$ in $w[1,r]$ in order to build $\Delta_a$. This is challenging
because the algorithm has not examined $w[1,r]$ yet. We can do it by traversing
the occurrence list $\Lambda_a$ of~$a$ in~$w$ (which was computed during the
preprocessing and maintained under letter
substitution updates, see \cref{sec:occlists})
and finding the $p-1$ smallest
$a$'s of~$w$. However, as $\Lambda_a$ is unsorted and has size
$|w|_a$ (which is not constant), this would take too long. For this reason, 
we must do this traversal \emph{in the background}.

More precisely,
at line~\ref{lin:minback}
we start a background
process which traverses each occurrence list $\Lambda_a$ for each $a \in \Sigma$. This is 
called the \emph{min background traversal for~$a$}: it is done by calling a
function \textsc{StartMinBackgroundTraversal} on~$\Lambda_a$, whose pseudocode we omit.
From then on, we have a constant number $|\Sigma|$ of background tasks 
that run in parallel with the rest of the code.
Namely, each min background
traversal for~$a$ retrieves one element of~$\Lambda_a$ at each instruction of
\cref{alg:antoineenum}, and memorizes the $\leq p-1$ smallest positions seen so
far:
this can be easily implemented in constant memory and constant
time per step because $p$ is constant.
Then, when $a$ becomes rare at $l=1$ (line~\ref{lin:getrareoccs}), then \textsc{GetRareOccurrencesSingle} 
 will retrieve the result of the min background traversal for~$a$. We
 illustrate in Appendix~\ref{apx:min} why the min background traversals are
 useful.

  The important point is that, whenever $a$ becomes rare at~$l=1$, then the
  min background traversal for~$a$ finishes in time. This is true because, at
  that point, 
  the right border $r$ has ``swept over'' all occurrences of~$a$ except at most
  $p-1$. %
  The formal claim is below and proved in Appendix~\ref{apx:finishmin}:

\begin{toappendix}
\subsection{Proof of Claim~\ref{clm:finishmin} On Min Background Traversals}
  \label{apx:finishmin}
\end{toappendix}

\begin{claimrep}
  \label{clm:finishmin}
  For $l=1$, for the first value of $r$ where a letter $a \in \Sigma$ becomes
  rare in $[1,r]$, then the min background traversal for~$a$ can be finished in
  constant time.
\end{claimrep}

\begin{toappendix}
  In light of the claim above, to retrieve the results of the min background traversal, we call
  the function \textsc{FinishMinBackgroundTraversal}($a$) in
  \cref{alg:getRareOccurrences}, whose implementation is omitted: this takes
  constant time, and it uses only the
  constant memory allocated in the min background traversal to remember the
  $\leq p$ first $a$'s. (Note that, deviating somewhat from what was explained
  in the main text of the paper, we use the min background traversal even in the
  case where a letter $a$ is initially rare for $l=1$; this amounts to the same as
  using the occurrence list $\Lambda_a$ directly, which is what we said in
  the main text for pedagogical reasons.)
\end{toappendix}

\begin{proof}[Proof of \cref{clm:finishmin}]
  When $a$ becomes rare (including if $a$ is initially rare), then
  the right endpoint $r$ has swept over all occurrences of~$a$ except at most
  $p$. Indeed, this is by definition of $a$ being rare in $[1,r]$.
  Now, each time $r$ was decremented, so in particular each time it has swept
  over an occurrence of~$a$, then we executed
  one step of the min background traversal for~$a$. Hence, when $a$
  becomes rare, we have processed all occurrences of $a$ in the occurrence list
  $\Lambda$ of~$a$ in~$w$ except at most $p$ occurrences. This ensures that we
  can indeed finish the min background traversal in constant time and retrieve
  the results.
\end{proof}

Thus, for $l=1$ we can obtain an occurrence list $\Delta_a$ of $a$ in $w[1,r]$ whenever $a$ becomes
rare. It is then easy to keep $\Delta_a$ up-to-date as $r$ gets decremented and
letter occurrences get removed
(line~\ref{lin:updatelists}).
This ensures that we can indeed complete the iteration $l=1$ of the for-loop: we
know when to exit the while-loop, and can produce remaining infixes in 
lines~\ref{lin:finish}--\ref{lin:finishend} using \cref{lem:easycase} %
because there are only finitely many non-neutral letters left
and we can process them in constant time.
This concludes the case $l=1$. (Note that the min background traversals are only
needed for $l=1$ and they will not be used in the sequel.)

\subparagraph*{Finding rare letters with \bm{$l>1$}.}  It remains to explain
how to locate the occurrences of rare letters for $l>1$, which is far more
challenging. To do this, we use
three ingredients: another kind of background
traversal of $\Lambda$ called \emph{max-left background traversals}, the notion of \emph{right information}, and the notion
of \emph{left information}. We now sketch these ingredients, with details
deferred to Appendix~\ref{apx:algodetails}.

First, the \emph{max-left background traversal} for each letter $a \in \Sigma$
is started at the end of the for-loop for $l=1$ (line~\ref{lin:maxback}). We
do it with a function \textsc{StartMaxLeftBackgroundTraversal} whose pseudocode
we omit. Again, these $|\Sigma|$ background processes run in parallel with
\cref{alg:antoineenum}, and they each traverse an occurrence list $\Lambda_a$
over the entire word. The definition of the max-left background traversals depends
on the value $r_1$, stored at line~\ref{lin:r1}, which is the
right position $r$ just before we exited the while-loop for~$l=1$; we call this the
\emph{limit for $l=1$}. (Note that we must have $r_1 \leq |w|$, because
if the body of the while-loop is not executed at all,
then we terminate immediately at lines~\ref{lin:nocond2} or~\ref{lin:allrare2},
so $r_1 = |w|+1$ is impossible.)
The max-left background traversal for~$a$
aims at locating the $\leq p$ last~$a$'s which are left of the limit~$r_1$.
We
define this
more precisely in Appendix~\ref{apx:algodetails}, and explain when we use the 
result of the max-left background traversals and why they finish in time. We
also illustrate in Appendix~\ref{apx:maxleft} why doing the max-left background
traversals is useful.

Second, the \emph{right information} (RI) at a position $r$ of~$w$ is a
constant-sized tuple describing some letter occurrences in~$w$ to the
\emph{right} of~$r$. We memorize this information as $r$ is decremented in
\cref{alg:antoineenum} in a data structure $\mathfrak{R}$, which is initialized
at every for-loop iteration (line~\ref{lin:initrs}) and
updated throughout the while-loop iteration (line~\ref{lin:updaters}): the
details are given in Appendix~\ref{apx:maintainliri}.
The
intuition is that, for the last right endpoint position $r_l \leq |w|$ before we
exited the while-loop for the
$l$-th iteration of the for-loop (we call $r_l$ the \emph{limit for~$l$}, generalizing $r_1$ from the previous
paragraph),
then the RI at $r_l$ contains some \emph{tracked} positions of~$w$ to the right of~$r_l$. In particular it
will include some positions that will be to the left of the position
$r_{l+1}$ where we will exit the while-loop for the next iteration of the
for-loop.
An illustration of the right information is in \cref{fig:RI} in
Appendix~\ref{apx:algodetails}.

\begin{definition}
  \label{def:ri}
  Let $w$ be a word and $r$ be a position of~$w$. The \emph{right 
  information} (RI) of $w$ at $r$ consists of the following for each letter $a
  \in \Sigma$:
  \begin{itemize}
    \item The sequence $r \leq \mu_{a,1} < \cdots < \mu_{a,n_a}$ of the $\leq p$ first
    $a$'s right of $r$
      (with $n_a \leq p$): we say that the positions $\mu_{a,1}, \ldots,
      \mu_{a,n_a}$ are
      \emph{tracked}.
    \item For each $1 \leq i \leq n_a$ and for each
      letter $b \in \Sigma$, the positions of the $\leq p$ last $b$'s left of the tracked
      position $\mu_{a,i}$ and right of~$r$.
  \end{itemize}
\end{definition}

The RI stores the first $p$ letters to the right of the current limit
(and their predecessors) because
the limit can \emph{jump} $p$ positions to the right
for one increment of the left border $l$.
The next example illustrates this.

\begin{example}
  \label{ex:limit}
  Consider the language $L_5$ on alphabet $\Sigma = \{a, b, c, e\}$
  that has $e$ as a neutral letter and
  contains the words satisfying one of the following:
  \begin{enumerate}
    \item it contains at least 5 $a$'s, at least 1 $b$ and at least 1 $c$,
      and contains a factor $b \Sigma^* c$,
    \item it contains at least 1 $a$, at least 5 $b$'s, and at least 1 $c$,
      and contains a factor $c \Sigma^* a$,
    \item it contains at least 1 $a$, at least 1 $b$, and at least 5 $c$'s,
      and contains a factor $a \Sigma^* b$,
    \item two of the letters $\{a,b,c\}$ occur at least 5 times
      and the remaining at least once.
  \end{enumerate}
  The language $L_5$ is ZG (intuitively because the order on letters only
  matters for rare letters) and it is extensible because each condition in the
  bullet points above is extensible. Hence, $L_5$ is in particular \threshext,
  and we can take threshold $p=5$.

  Consider now the following input word $w$ with neutral letters already omitted.
  \begin{alignat*}{22}
    index:  &  & 1&& 2&& 3&& 4&& 5&& 6&& 7&& 8&& 9&&10&&11&&12&&13&&14&&15&&16\\
    w & =&\quad a&&\quad b&&\quad b&&\quad b&&\quad a&&\quad a&&\quad a&&\quad
    a&&\quad c&&\quad c&&\quad c&&\quad c&&\quad b&&\quad c &&\quad a &&\quad a
  \end{alignat*}

  Note that for words where the letter $a$ is frequent ($|w|_a\geq 5$),
  the language $\cond(\{a\})$ contains all words
  that contain at least one $b$ and one $c$ and where the $b$ occurs before the $c$.
  This condition $\cond(\{a\})$ is valid for $w[1,9]$,
  but no longer valid for $w[1,8]$ as there are no $c$.
  Thus, the first limit will be at position $r_1=9$.

  For the left border $l=2$, when $r\leq 14$, the word now only contains 4 $a$'s.
  The smallest factor $w[2,r]$ that contains at least a frequent letter is $w[2,14]$.
  There, the letter $c$ is frequent
  and $\cond(\{c\})$ contains all words with at least one $a$ and one $b$
  and where the $a$ occurs before the $b$.
  Thus, $\cond(\{c\})$ is fulfilled and the second limit is $r_2=14$.
  Note that position $r_2$ is the fifth $c$ of the $\leq 5$ first $c$'s right of
  $r_1$ (and it is not $|w|+1$).
  So this illustrates why the RI has to store so many positions.
\end{example}

We now explain the third and final ingredient.
For $r\geq r_1$ with $r_1$ the limit for~$l=1$,
the \emph{left information} (LI) at a position $r$ of~$w$
is a
constant-sized tuple of letter occurrences in~$w$ to the
\emph{left} of~$r$, but to the right of~$r_1$.
(Intuitively, the positions to the left of~$r_1$ are handled by the
max-left background traversals.)

\begin{definition}
  \label{def:li}
  Let $r_1$ be the position where the while-loop exits for~$l=1$.
  The \emph{left information} (LI) of $w$ at the limit $r_l$ (for $l\geq 1$),
  contains,
  for each letter $a \in \Sigma$, the $\leq p$ last $a$'s left of $r_l$ and
  right of $r_1$.
\end{definition}

The LI will only be used for values of~$r$ that are limits,
i.e., where we exit the while-loop. For $l=1$, we initialize the LI at $r_1$
at the end of the outer for-loop ($l=1$) at line~\ref{lin:initls}: this is easy because the LI for $l=1$ concerns positions of~$w$ to the left of~$r_1$ and to the right of~$r_1$ (i.e., only in~$r_1$). Then, we update
the LI at~$r_l$ at the end of the $l$-th
for-loop iteration (at line~\ref{lin:updatels}) from the LI at
$r_{l-1}$ and the RI at~$r_{l-1}$: the details are again in
Appendix~\ref{apx:maintainliri}.

All that remains is to show how rare letter occurrences can be
properly computed whenever a letter becomes rare for $l>1$. This is the object
of the claim below (see Appendix~\ref{apx:algodetails}):

\begin{claim}
  \label{clm:algofinish}
  By maintaining the RI at the successive values of~$r$, the RI at the limit
  $r_{l-1}$ where we exited the for-loop at $l-1$, and the LI at~$r_{l-1}$,
  together with the results of the max-left background traversals (when
  available), we can determine in constant time,
  for each $l>1$ and letter $a\in\Sigma$,
  an occurrence list $\Delta_a$ of the occurrences of~$a$ in $w[l,r]$ whenever
  $a$ becomes rare (i.e., either initially at line~\ref{lin:getrareocc}, or
  in the while-loop at line~\ref{lin:getrareoccs}).
  This can be done while obeying the constant delay and
  constant enumeration memory requirements.
\end{claim}

This concludes the presentation of \cref{alg:antoineenum},
and finishes the proof of \cref{thm:main}.

\begin{toappendix}

  \subsection{Additional Details and Proof of Claim~\ref{clm:algofinish}}
  \label{apx:algodetails}

  This appendix gives more details on the main algorithm towards proving \cref{clm:algofinish}, which is all that is needed to
  establish \cref{thm:main}. It is
  subdivided in four sections. First, we give a formal
  explanation of the terminological choices used to refer to letter
  occurrences. %
  Second, 
  we formally define the notion of \emph{limit} which was used informally in the main text,
and characterize
how the limit can move from one value of $l$ to the next.
Third, we explain more precisely how the right information (RI) and left information (LI)
can be computed and
maintained efficiently.
  Then, we 
  prove \cref{clm:algofinish} by explaining how the rare
  letter occurrences can be obtained from the max-left background traversals and
  the RI and LI.

\subsubsection{Terminological Choices}
In the presentation of the algorithm, and in some definitions 
(e.g., \cref{def:ri} and \cref{def:li}), we often talk about the
occurrences of non-neutral letters that are to the left or right of a specific
word position. So we now formally define the terminology that we use.

Let $1 \leq i \leq j \leq |w|$ be positions of the word $w$.
By the \emph{$\leq p$ last $a$'s  left of~$j$ and right of~$i$}, we denote
the (possibly empty) ordered sequence of positions $(l_q, \ldots, l_1)$ with
$q\leq p$ maximal and with $i\leq l_q < \cdots < l_1 \leq j$, each $l_i$ being
  maximal given the previous ones, such that (1.) all positions contain
$a$'s and (2.) the positions correctly reflect the last $\leq p$ $a$'s of
$w[i,j]$ and reflect all $a$'s if there are less than~$p$. 
Formally:
\begin{enumerate}
  \item For all $1 \leq k \leq q$, we have $w[l_k] = a$
  \item Letting $i' \coloneq i$ if $q < p$ and $i' \coloneq l_q$
    if $q = p$, then the
    following holds:
for all $i' \leq x \leq j$, if $x \notin \{l_q, \ldots, l_1\}$ then
        $w[x] \neq a$
    \end{enumerate}
Note that $q = \min(|w[i,j]|_a, p)$.
In particular, if there are no $a$'s in $w[i,j]$, then $q=0$
and the sequence is empty.

When $i=1$ we omit
$i$ from the definition and simply talk of the \emph{$\leq p$ last $a$'s left of~$j$}. 
When $j=|w|$ we omit $j$ from the definition and simply talk of the \emph{$\leq
p$ last $a$'s right of~$i$}.
Further when additionally
$j=|w|$ we omit $j$ and talk of the \emph{$\leq p$ last $a$'s}.

Symmetrically, we
define the \emph{$\leq p$ first $a$'s right of~$i$ and left of~$j$} 
in the expected symmetric way: they form a sequence $i \leq l_1 < \cdots <
l_q \leq j$ with $q \leq p$ as large as possible and every $l_i$ as small as
possible. Formally:
\begin{enumerate}
  \item For all $1 \leq k \leq q$, we have $w[l_k] = a$
  \item Letting $j' \coloneq j$ if $q < p$ and $j' \coloneq l_q$
    if $q = p$, then the
    following holds:
for all $i \leq x \leq j'$, if $x \notin \{l_1, \ldots, l_q\}$ then
        $w[x] \neq a$
    \end{enumerate}
    Likewise, we define 
the \emph{$\leq p$ first $a$'s right of~$i$}
when $j = |w|$ and the \emph{$\leq p$ first $a$'s} when additionally $i=1$.

  \subsubsection{Limits}
Now, let us give the formal definitions for the notions of \emph{limits}. Recall our global
description of the algorithm. For each value of $l$ from $1$ to the last value
considered, the algorithm will consider
all values of $r$ by decreasing order starting from $|w|$ until the while-loop
is exited, i.e., either the set
$T$ of frequent non-neutral letters of $w[l,r]$ becomes empty (\allrare)
or $w[l,r]$ is no longer in~$\cond(T)$ (\nocond).
The value $r+1$, corresponding to the right endpoint of the last infix outputted
in line~\ref{lin:output},
will be called the \emph{limit} for~$l$ and is written
$r_l$.
In the case where the while-loop terminates immediately (and its body did not
execute for that value of~$l$ and did not output any infix), the limit is $r_l
\coloneq |w|+1$: however this case never matters because we then
terminate immediately at lines~\ref{lin:nocond2} or~\ref{lin:allrare2}.

Remember that,
by \cref{threshdef-condp},
membership to
$\cond(T)$ for the word $w[l,r_l]$ only depends on
$w[l,r_l]_{-(N\cup T)}$ (with $N$ the neutral letters for~$L$ in~$\Sigma$),
i.e., it can be checked by considering only the occurrences of rare letters
given by~$\Lambda$. Then
the following is true by definition:
\begin{claim}
  \label{clm:nostop}
  For each $l \geq 1$ where $r_l < |w|+1$,
  the set $T$ of frequent letters of $w[l,r_l]$ is nonempty, and
  we have $w[l,r_l] \in \cond(T)$.
  In particular $w[l,r_l] \in L$.
\end{claim}

It will be important in the algorithm to understand how the limit changes
between one left endpoint position $l$ and the
next left endpoint position $l+1$ (i.e., the next iteration of the for-loop):

\begin{lemmarep}
  \label{lem:boundedjump}
  For any $1 \leq l < |w|$ such that the algorithm runs the $(l+1)$-th iteration
  of the for-loop (i.e., has not terminated yet), we have that the limit
  $r_{l+1}$ satisfies one of the following conditions:
  \begin{itemize}
    \item it is equal to $r_l$;
    \item it is equal to $|w|+1$ (i.e., we terminate immediately for the
      $(l+1)$-th iteration);
    \item it is one of the $\leq p$ first $a$'s right of~$r_l$ for some
  non-neutral letter $a$.
   \end{itemize}
\end{lemmarep}

It it rather intuitive why the limit $r_{l+1}$ can be part of the $\leq 2$ first
$a$'s right of~$r_l$, e.g., for $a \coloneq w[l]$ the letter that ``exits'' at $l$. (Remember
that the $\leq 2$ first $a$'s right of~$r_l$ may include $r_l$ if we also have
$w[r_l] == a$, which is why we say ``$\leq 2$'' and not ``$\leq 1$''.)
But it may not be intuitive why we need to
consider the $\leq p$ first $a$'s. Refer back to \cref{ex:limit} which
illustrates the possible cases of
\cref{lem:boundedjump}.

To show \cref{lem:boundedjump}, we first make an observation (similar to \cref{lem:condmonotone})
on the monotonicity of the Cond
predicates from \cref{threshdef-condp}, where we refer to the notion of rare and
frequent letters that depends on the language $L$ and threshold $p$:

\begin{claim}
  \label{clm:condmon}
  Let $u \in \Sigma^*$ be a word and let $T$ be its set of frequent letters. Let
  $u'$ be an extension of $u$ and let $T' \supseteq T$ be its set of frequent letters. If $u
  \in \cond(T)$, then $u'\in \cond(T')$.
\end{claim}

\begin{proof}
  As $u \in \cond(T)$, there exists a word $v \in L$ ensuring that $v_{-T}$ is a
  factor of $u_{-T}$. Now, as $u'$ is an extension of~$u$, we have $T' \supseteq
  T$. Thus, we also have that $v_{-T'}$ is a factor of $u_{-T'}$. As $u$ is a
  factor of~$u'$ we have that $v_{-T'}$ is a factor of $u'_{-T'}$. Thus by
  definition $u' \in \cond(T')$.
\end{proof}

We can now prove \cref{lem:boundedjump}:

\begin{proof}[Proof of \cref{lem:boundedjump}]
  If $r_{l+1} = |w|+1$ then we have nothing to show, so in the rest of the proof we
  assume that $r_{l+1} \leq |w|$.

  Let us first show that $r_{l+1}\geq r_l$. Assume by contradiction that
  instead $r_{l+1}< r_l$.
  Consider the last infix $w_{l+1} \coloneq w[l+1,r_{l+1}]$ outputted
  before exiting the while-loop at the $(l+1)$-th iteration of the for-loop. 
  It must be the case by \cref{clm:nostop} that the set of frequent letters of
  $w_{l+1}$ is non-empty and $w_{l+1}$ is
  in $\cond(T)$ for the set $T$ of its frequent letters.
  Consider now the factor $w_l'\coloneq w[l,r_l-1]$, with $r_l-1 \geq r_{l+1}$: clearly $w_l'$ is an
  extension of~$w_{l+1}$.
  Now, by \cref{clm:condmon}, we have that $w_l'$ is in $\cond(T')$ for the
  non-empty set $T' \supseteq T$ of its frequent letters. So at the $l$-th
  iteration of the for-loop we should not have exited the while-loop at $r_l$,
  because $r_{l-1}$ still satisfied the conditions. This is a contradiction to
  the definition of~$r_l$.

  Let us now show that $r_{l+1}$ falls in one of the cases from the lemma
  statement. Consider the letter $a$ occurring at position $l$. For any $l+1 \leq r
  \leq |w|$, we have that $w[l+1,r]$ is the same as $w[l,r]$ except it is
  missing an $a$ at the beginning. There are two cases depending on why 
  we left the inner while-loop for $l+1$. First, if it is
  because of \allrare, then either the last letter that becomes rare is
  different from~$a$ and we have $r_{l+1} = r_l$, or that letter is $a$. In that
  case, as $p\geq 2$, we have that $r_{l+1}$ is in the $\leq 2$ first $a$'s
  right of~$r_l$. (In particular if there are
  no $a$'s at a position strictly greater than~$r_l$ then $r_{l+1}$ is~$|w|+1$,
  but we already excluded that case.)

  Second, if it is because of
  $\cond(T')$ for $T'$ the set of letters that are frequent in
  $w[l+1,r_{l+1}]$,
  we reason as follows. 
  Remembering that we assumed $r_{l+1} \leq |w|$, the
  position $r_{l+1}$ must ensure that either the set
  $T$ of letters frequent in $w[l+1,r_{l+1}-1]$ is different from the set $T'$ of
  letters frequent in $w[l+1,r_{l+1}]$; or the two sets are the same but
  $w[l+1,r_{l+1}]_{-T}$ is in $\cond(T)$ but $w[l+1,r_{l+1}-1]_{-T}$ is not.

  In the first subcase, it must be that $T \subseteq T'$ and $T' \setminus T$ is
  one single letter, namely, the letter $a \coloneq w[r_{l+1}]$, which is
  non-neutral.
  The letter $a$ became rare, so we have (*):
  $|w[l+1,r_{l+1}-1]|_a < p$.
  Now, at least a letter must be frequent in $w[l,r_l]$ by \cref{clm:condmon}
  and we chose $p \geq 2$, so we know that $r_l > l$, thus $l+1 \leq r_l$. From
  (*) we deduce that 
  $|w[r_l,r_{l+1}-1]|_a < p$.
  Now, observe that $w[r_{l+1}]$ contains $a$ by definition. So indeed $r_{l+1}$
  is part of the $\leq p$ first $a$'s right of~$r_l$.

  In the second subcase, the words $w[l+1,r_{l+1}]_{-T}$ and
  $w[l+1,r_{l+1}-1]_{-T}$ must differ by a letter which is non-neutral and not
  in~$T$, so the letter $a \coloneq w[r_{l+1}]$ must be a letter which is rare
  in $w[l+1,r_{l+1}]$ and in $w[l+1,r_{l+1}-1]$, i.e., $|w[l+1,r_{l+1}]|_a < p$.
  We then conclude exactly like in the first subcase.
\end{proof}

  \subsubsection{Computing and Maintaining the LI and RI}
  \label{apx:maintainliri}
  In this appendix, we give the precise explanation of how the left information
  (LI) and right
  information (RI) can be computed on the input word and how it can be maintained as
  the left and right endpoints change. We then explain in the next section
  how this information,
  together with the occurrence counts described earlier, and with the max-left
  background traversals, can be used to find occurrences of rare letters.

Recall that the RI is initialized at every for-loop
iteration (line~\ref{lin:initrs}; also including the first iteration $l=1$ of
the for-loop) with $r=|w|$,
and updated throughout the while-loop iteration by decrementing~$r$
(line~\ref{lin:updaters}).
The LI is initialized at the end of the
outer for-loop ($l=1$) at line~\ref{lin:initls} at the right limit~$r_1$,
and updated at the end of each
for-loop iteration $l+1$ (at line~\ref{lin:updatels}) from the LI
at the previous limit $r_l$ (computed at the $l$-th for-loop iteration) and from the 
RI at the previous limit $r_l$ (computed at the end of the $l$-th
for-loop iteration).

\begin{figure}
  \begin{tikzpicture}
    \node at (-0.3,0) (w) {$\boldsymbol w$};
    \draw[thick] (0,0) -- (12,0);
    \draw[thick] (0,.5) -- (0,-.5);
    \node at (0,1) (1) {$1$};
    \draw[thick] (12,.5) -- (12,-.5);
    \node at (12,1) (lengthw) {$|w|$};
    \draw (1,.2) -- (1,-.2);
    \node at (1,-.5) (r) {$\boldsymbol r$};
    \foreach \x in {1,2,3.5} {
      \tikzmath{\a = .5+2.5*\x;}
      \draw (\a,.2) -- (\a,-.2);
      \node at (\a,-.5) (r) {\pgfmathparse{ifthenelse(\x<3,"$\mu_{a,\x}$","$\mu_{a,n_a}$")} \pgfmathresult};
      \node at (\a,.5) (a) {$a$};
      \draw [decorate,decoration={brace,amplitude=5pt,mirror}]
      (\a-1.2,-.8) -- (\a-.1,-.8) node[midway,below=2mm,align=center]{$\leq p$ last $b$'s\\ left of \pgfmathparse{ifthenelse(\x<3,"$\mu_{a,\x}$","$\mu_{a,n_a}$")} \pgfmathresult\\ and right of $r$};
    }
    \node at (.5+2.5*2.75,.5) (dots) {$\cdots$};
    \node at (.5+2.5*2.75,-.5) (dots) {$\cdots$};
    \draw [decorate,decoration={brace,amplitude=5pt}]
    (.4+2.5*1,.8) -- (.6+2.5*3.5,.8) node[midway,above=2mm]{$\leq p$ first $a$'s right of $r$};
  \end{tikzpicture}
  \centering
  \caption{Right information: for every $a\in\Sigma$ and then for every
  $b\in\Sigma$.}
  \label{fig:RI}
\end{figure}
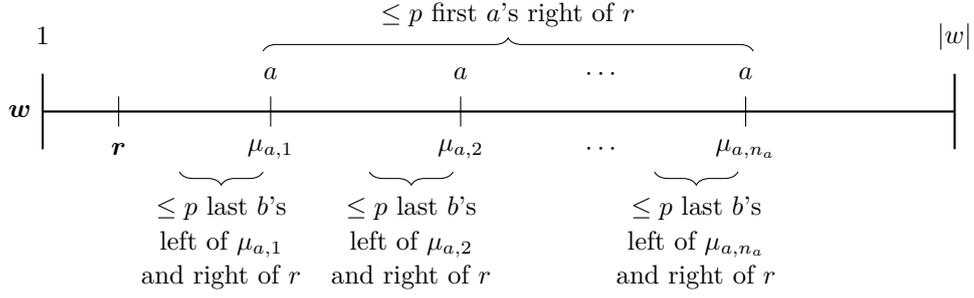

\subparagraph*{RI computation.}
For every iteration of the for-loop, the algorithm stores a constant-sized
data structure $\mathfrak{R}$ storing the RI at the current
right endpoint position~$r$. 
The structure $\mathfrak{R}$ can be trivially initialized for $r = |w|$,
and it can be easily maintained at every step:

\begin{proposition}
  \label{prp:ri}
  We can build in constant time a data structure $\mathfrak{R}$ 
 containing the
  (trivial) right information at $r = |w|$, and keep it up-to-date in constant
  time when decrementing $r$ to contain the right information at
  $r$. Further, $\mathfrak{R}$ is always of constant size.
\end{proposition}

In the algorithm, we initialize $\mathfrak{R}$ by a function \textsc{InitRI}()
and update it with a function \textsc{UpdateRI}($\mathfrak{R}$, $r$).
For both these functions, we do not give the pseudocode.

\begin{proof}
  For the initialization, the right information at $r = |w|$ consists simply
  tracking the occurrence at position $|w|$ of the letter at $w[|w|]$ (if it is
  non-neutral), which is easy to compute in constant time.

  For the update, as $r$ gets decremented, the occurrence at $r$ of 
  letter $a \coloneq w[r]$ (if it is non-neutral) gets stored as a tracked position,
  i.e., as the first $a$ to the right of~$r$. (In exchange, if
  we are already tracking $p$ occurrences of $a$ to the right of~$r$, then we
  stop tracking the last such occurrence.)
  Second, for each tracked $a \in \Sigma$ at position $r_a$, if we
  do not already store $p$ occurrences of $b \coloneq w[r]$ to the left of $r_a$,
  then we add $r$ as an occurrence of $b$ to the left of~$r_a$ (and to the right
  of $r$). 
  It is immediate that all of this can be implemented in constant time on a
  data structure that uses constant memory overall.
\end{proof}

\subparagraph*{LI computation.}
We maintain a data structure $\mathfrak{L}$ that stores the LI at $r_l$.
Initially, we compute $\mathfrak{L}$ for $r_1$ at the end of the first iteration of the
for-loop, which is trivial (i.e., we only store the occurrence at $r_1$ of
the letter $w[r_1]$ if it is non-neutral). 
This is done by \textsc{InitLI}($r_1$) whose pseudocode we do not give.
In subsequent iterations, we will compute $\mathfrak{L}$ for the current
value of $r_l$ from the previous structure $\mathfrak{L}_{-1}$
containing the LI at $r_{l-1}$ and the previous structure
$\mathfrak{R}_{-1}$ containing the RI at $r_{l-1}$. Formally:

\begin{claimrep}
  \label{clm:maintainls}
  For each value $l>1$ in the outer for-loop, when the inner while-loop
  concludes (i.e., we exit the loop and we have $r = r_l-1$), then either we
  finish immediately (i.e., $r = |w|$),
  or we can compute the data structure $\mathfrak{L}$
  storing the LI at $r_l$ in constant time 
  at line~\ref{lin:updatels} from the following:
  \begin{itemize}
    \item The data structure $\mathfrak{L}_{-1}$ storing the LI at $r_{l-1}$,
    \item The data structure $\mathfrak{R}_{-1}$ storing the RI at $r_{l-1}$.
  \end{itemize}
\end{claimrep}
The computation above is done by the function \textsc{UpdateLI}($\mathfrak{L}_{-1}$, $\mathfrak{R}_{-1}$) in the algorithm, whose pseudocode we omit.

The above uses the following immediate claim.

\begin{claim}
  \label{clm:leftright}
  For any positions $1 \leq r_1 \leq i < j \leq |w|$ of~$w$, for any letter~$a$,
  given the $\leq p$ last $a$'s left of $i$ and right of~$r_1$, and given the 
  $\leq p$ last $a$'s left of $j$ and right of~$i$,
  we can compute in constant time the $\leq p$ last $a$'s left of~$j$ and right
  of~$r_1$.
\end{claim}

We can now prove \cref{clm:maintainls}. Note that this is the only place where
we use the second point in the definition of the RI, namely, that the RI stores
letter occurrences to the left of the tracked positions.

\begin{proof}[Proof of \cref{clm:maintainls}]
  By \cref{lem:boundedjump}, we can distinguish three cases for the limit $r_l$.

  The first case is where $r_l = r_{l-1}$. In which case, we trivially conclude
  by setting $\mathfrak{L} \coloneq \mathfrak{L}_{-1}$.

  The second case is where $r_l$ is one of the positions tracked by the RI at
  $r_{l-1}$. Thus, by definition of the RI (second point),
  for each letter $b \in \Sigma$, the structure
  $\mathfrak{R}_{-1}$ storing the RI at $r_{l-1}$ gives us
  the following (*): the $\leq p$ last $b$'s left of the tracked position~$r_l$
  and right of~$r_{l-1}$. Further, for each letter $b \in \Sigma$ the structure
  $\mathfrak{L}_{-1}$ stores the following ($\dagger$): the $\leq p$ last $b$'s left
  of~$r_{l-1}$ and right of~$r_1$. Putting (*) and ($\dagger$) 
  together using \cref{clm:leftright} for each letter $b \in
  \Sigma$, we obtain for each $b \in \Sigma$ the $\leq p$
  last $b$'s left of~$r_l$ and right of~$r_1$. This gives us the LI at $r_l$
  which is what we needed.

  The third case is where $r_l = |w|+1$, i.e., we exit the while-loop
  immediately. But this was excluded from the outset: the algorithm
  finishes immediately before reaching line~\ref{lin:updatels}.
\end{proof}

\subsubsection{Finding Rare Letter Occurrences: Concluding the Proof of \cref{clm:algofinish}}

\begin{algorithm}[t]
\caption{GetRareOccurrences}
  \label{alg:getRareOccurrences}
\begin{algorithmic}[1]
\Procedure{GetRareOccurrencesSingle}{$a, l, r, \mathfrak{L}_{-1}, \mathfrak{R}_{-1}$}
  \Statex ~~~~\,\textbf{Input:} Letter $a$, positions $l$ and $r$ with $|w[l,r]|_{a} < p$,
  structure $\mathfrak{L}_{-1}$ storing the LI at
  \Statex \hspace{\algorithmicindent} $r_{l-1}$ (or unset if
  $l=1$), structure $\mathfrak{R}_{-1}$ storing the RI at $r_{l-1}$ (or
  unset if $l=1$)
  \Statex ~~~~\,\textbf{Output:}
  $\Delta[a]$ occurrence lists with all occurrences of $a$ in
  $w[l,r]$
  \If{$l == 1$}
    \Comment{$\mathfrak{L}$ and $\mathfrak{R}$ are unset}
    \State $\Delta[a] \gets$ \Call{FinishMinBackgroundTraversal}{$a$}
    \label{lin:getmin}
    \Comment{see \cref{clm:finishmin}}
  \Else
    \State $\Delta' \gets$ \Call{FinishMaxLeftBackgroundTraversal}{$a$}
    \label{lin:maxleftfinishes}
    \Comment{see \cref{clm:rareincomplete}}
    \State $\Delta[a] \gets$ \Call{GetRare}{$\mathfrak{L}_{-1}$,
    $\mathfrak{R}_{-1}$, $l$, $r$, $\Delta'$}
    \label{lin:getrare}
    \Comment{see \cref{clm:rareincomplete}}
  \EndIf
  \State \Return $\Delta[a]$
\EndProcedure
\Statex
\Procedure{GetRareOccurrences}{$l, r, \mathfrak{L}_{-1}, \mathfrak{R}_{-1}, \text{nocc}$}
  \State $\Delta \gets$ empty occurrence list
  \ForAll{$a \in \Sigma$}
    \If{$\text{nocc}[a] < p$} \Comment{letter $a$ is rare}
      \State $\Delta[a] \gets$ \Call{GetRareOccurrencesSingle}{$a$, $l$, $r$, $\mathfrak{L}_{-1}$, $\mathfrak{R}_{-1}$}
    \EndIf
  \EndFor
  \State \Return $\Delta$
\EndProcedure
\end{algorithmic}
\end{algorithm}

We conclude by finally explaining how the LI and RI allow us to find rare letter
occurrences and test the condition of the while-loop for $l>1$, together with
the results of the max-left background traversal, to finish the proof of 
\cref{clm:algofinish}. This is presented in
\cref{alg:getRareOccurrences}, whose functions are called at lines
\ref{lin:getrareocc} and \ref{lin:getrareoccs} in \cref{alg:antoineenum}.

The code of \cref{alg:getRareOccurrences} also handles the case $l=1$ in the way
that was already explained in the main text. For the case $l>1$,
the only interesting lines of \cref{alg:getRareOccurrences} are
line~\ref{lin:maxleftfinishes} and line~\ref{lin:getrare}. Recall that, as
explained in the main text,
the max-left background traversal gives us for all
$a \in \Sigma$ the $\leq
p$ $a$'s left of~$r_1$, for $r_1$ the limit for~$l=1$. We
use a function \textsc{FinishMaxLeftBackgroundTraversal}($a$) (whose pseudocode
we omit) to retrieve this at line~\ref{lin:maxleftfinishes}. (Note that this can
be called for several values of~$l$, so the result of the max-left background
traversal, once finished, is stored and can be used in the rest of the execution
of the algorithm; but for each $a \in \Sigma$ there is only one single max-left
background traversal that will ever be started and completed.)
We then combine this result with data structures storing the LI and RI, at
line~\ref{lin:getrare}, using a function \textsc{GetRare} whose pseudocode we
omit.

Only two points remain to be shown: that the max-left background traversal
finishes in time for when we want to use its results, and that the occurrences
of a letter can indeed be computed from the LI and RI and the results of the
max-left background traversal. We show these in a combined lemma:

\begin{claimrep}
  \label{clm:rareincomplete}
  For each $l > 1$, let $r$ be the minimal
  position where a non-neutral letter $a \in
  \Sigma$ is frequent in $w[l,r]$ (and set $r = |w|+1$ if $a$ is initially
  rare in $w[l,|w|]$). Then:

  \begin{itemize}
    \item We can finish the max-left background traversal for~$a$ in
      constant time (if not already finished) and obtain from its results
      the $\leq p$ last $a$'s left of~$r_1$.
    \item From the $\leq p$ last $a$'s left of~$r_1$, together with the LI at
      $r_{l-1}$ and RI at $r_{l-1}$, we can obtain the $<p$ occurrences of~$a$
      in $w[l,r-1]$.
  \end{itemize}
\end{claimrep}

\begin{proof}
  We first show the first point.
  As $a$ is rare in $w[l,r-1]$, we have $|w[l,r-1]|_a< p$. We know that all
  occurrences of~$a$ in~$w$ to the right of~$r$ have
  been swept over by the while-loop (at least for the execution for the current
  value of~$l$); and we know that all occurrences of~$a$ to the left
  of~$l$ have been swept over by the left endpoint with $l>1$ in the executions
  of the for-loop.
  Hence, since the max-left background traversal started, we have swept over all
  occurrences of~$a$ except at most~$p$. So we can terminate that process in
  constant time (if not yet done). The result of the max-left background
  traversal gives us the $\leq p$ last $a$'s left of~$r_1$.

  We now show the second point. The key point is that it suffices to explain how we can
  obtain the $\leq p$ last $a$'s left of~$r$. Indeed, they necessarily cover all
  $a$'s in $w[l,r-1]$, as there are at most $p-1$~-- so this is what we explain in
  the sequel.

  We have obtained from the max-left background
  traversal the $\leq p$ last $a$'s left of~$r_1$. The LI at~$r_{l-1}$ gives us
  the $\leq p$ last $a$'s left of~$r_{l-1}$ and right of~$r_1$.
  Concatenating and truncating these lists, similarly to \cref{clm:leftright},
  gives us the following (*): the $\leq p$ last $a$'s left of~$r_{l-1}$. Now, by
  the first point of \cref{def:ri},
  the RI at~$r_{l-1}$ gives us the following ($\dagger$):
  the $\leq p$ first $a$'s right of~$r_{l-1}$ for each
  $a \in \Sigma$ (the
  tracked positions).

  Now we must reason about what the value of~$r$ can be, relative to the
  previous limit $r_{l-1}$. (The reasoning that follows resembles
  \cref{lem:boundedjump}, but note that it applies to a position $r$ which is
  the minimal position where $a$ becomes rare but which is not necessarily the
  limit $r_l$.)

  We first assume $r < |w|+1$, and deal with the case
  $r = |w|+1$ (i.e., $a$ initially rare) at the end. As $a$ is rare in
  $w[l,r-1]$ but not in $w[l,r]$, we know that $w[r] = a$, and
  $|w[l,r-1]|_a<p$, so $r$ is one of the $\leq p$ first $a$'s right of~$l$.
  By our initial choice of $p \geq 2$, we know that $l-1 < r_{l-1}$ so that $l
  \leq r_{l-1}$, so we also have $|w[r_{l-1},r-1]|_a<p$, so that $r$ is one of
  the $\leq p$ first $a$'s right of~$r_{l-1}$.
  Thus, by ($\dagger$), $r$ is a position tracked by the RI at~$r_{l-1}$,
  and in fact
  the RI contains 
  the $\leq p$ last $a$'s left of~$r$
  and right of $r_{l-1}$, because all these positions are tracked. (Note that,
  here, we only use the first point
  of the definition of the RI in \cref{def:ri} -- we don't use the second point.)
  Concatenating this with (*), we obtain the $\leq p$ last $a$'s
  left of $r$, which is what we needed.

  We now discuss the case where $r=|w|+1$, i.e., $a$ is initially rare. We make
  the claim (\#): 
  there are $\leq p$ occurrences of~$a$ right of~$r_{l-1}$ in~$w$. Indeed, assume
  by way of contradiction that there are at least $p+1$. Then, there are at
  least $p$ occurrences of~$a$ right of~$r_{l-1}+1$. Like in the previous paragraph
  we have $l \leq r_{l-1}$, so $l \leq r_{l-1}+1$.
  But then $w[l,|w|]$ would contain at least these $p$
  occurrences of~$a$,
  which contradicts the fact that $a$ is rare in $w[l,|w|]$. So
  indeed we have established (\#). From (\#), we know that all occurrences of~$a$
  right of~$r_{l-1}$ are tracked by the RI at $r_{l-1}$ (again by the first
  point.). Combining these with
  (*), we obtain the $\leq p$ last $a$'s left of~$|w| = r-1$, which is again
  what we needed and which concludes.
\end{proof}

Thus, we have explained how the algorithm maintains the RI and LI, and how
they can be used to test the condition of the inner while-loop. As we explained
in our overall presentation of the algorithm,
this ensures that the algorithm correctly enumerates all $L$-infixes (using
\cref{clm:nostop} to argue that the infixes in the main algorithm are actually
$L$-infixes, together with \cref{lem:easycase}), and
concludes the proof of 
\cref{clm:algofinish} 
and hence of \cref{thm:main}.

\subsection{Examples for Background Traversals}
\label{apx:exabackground}

We last give some examples illustrating why the notion
of \emph{background traversal} used in the proofs is needed, by illustrating 
each kind of background traversal in turn. 
This material is only meant to convey intuition and is not part of the formal
proof of our results.

\subsubsection{Example Min Background Traversal}
\label{apx:min}

The following example illustrates why the min background traversal is
needed:

\begin{example}
  Let $L_{ab}$ be the language from \cref{exa:threshext}.
  Now consider the following input word $w$
  where we do not specify the amount of neutral letters $e$.
  \begin{alignat*}{20}
    index:  &  & && i_1&& && i_2&& && i_3&& && i_4&& && i_5&& && i_6&& && i_7 && && i_8 &&\\
    \boldsymbol{w} & \boldsymbol{=}&\quad e^* &&\quad \boldsymbol{a}&&\quad e^*&&\quad \boldsymbol{b}&&\quad e^*&&\quad \boldsymbol{a}&&\quad
    e^*&&\quad \boldsymbol{b}&&\quad e^*&&\quad \boldsymbol{a}&&\quad e^*&&\quad \boldsymbol{b}&&\quad e^*&&\quad \boldsymbol{a}&&\quad e^*&&\quad \boldsymbol{b}&&\quad e^*&&
  \end{alignat*}

  The algorithm starts with occurrence lists $\Lambda_a$ and $\Lambda_b$
  where the indices $i_j$ are classified by letter but are in random order.
  At the start, the min background traversal for every letter is started.

  In this example, we only look at the first iteration of the for-loop ($l=1$).
  At the beginning, the occurrence lists of rare letters $\Delta$ is empty
  and the set of frequent letters $T$ contains both letters $a$ and $b$.
  The condition in line~\ref{lin:whilecond} of the algorithm is easily tested,
  because the word $w[1,r]_{-T}=e^*$ is contained in $\cond(\{a,b\})=\Sigma^*$.
  Whenever $r$ is decremented,
  the min background traversal of every letter moves forward by one step.

  In the body of the inner while-loop, when $r=i_6$, in line~\ref{lin:noccrdec},
  the counter gets decremented to $\text{nocc}'[b]=2$ and
  is thus smaller than $p=3$. We set $T=\{a\}$
  and then, the algorithm calls \textsc{GetRareOccurrencesSingle}($b$, 1, $i_6-1$)
  (remember that $\mathfrak{L}$ and $\mathfrak{R}$ are not used in this case $l=1$).
  At this point, the min background traversal can be finished in constant time,
  by \cref{clm:finishmin}.
  This is possible because the traversal already read at least one letter
  at every position $r\geq i_6$ and there are only $p-1$ many $b$ left,
  thus, the traversal can be finished in constant time.
  It stores $\Delta[b]=\{i_2, i_4\}$.

  At position $r=i_5$, the same happens to the letter $a$
  and we store $\Delta[a]=\{i_1, i_3\}$.

  Directly after decrementing $r$ to $i_5-1$, the while-loop exits
  because $T=\emptyset$ (\allrare).

  The conditional expressions after the while-loop
  call \textsc{ProduceRemainingL}(1, $i_5-1$, $\Delta$).
  With the constant amount of data stored in $\Delta$ (only two indices per letter),
  it is now easy to decide if the remaining word contains $a b$
  and to enumerate the remaining infixes for the fixed left border $l=1$:
  here we will output infixes $[1,j]$ for $i_2\leq j<i_3$.

  Note that the algorithm never touches the indices $i_1$ nor $i_2$
  with the pointers $l$ or $r$ at that point. So the min background traversal is
  important to identify the positions of $a$ and $b$ in $w[1,i_5-1]$. Indeed, if
  the initial word had been different and had a letter $b$ at position $i_1$ and an $a$ at
  position $i_2$, then there would be no infixes of the form $[1,j]$ to produce
  with $1 \leq j < i_5$; and this word would be indistinguishable from the word on
  which we ran the example, were it not for the results of the min background
  traversals.
\end{example}

\subsubsection{Example Max-Left Background Traversal}
\label{apx:maxleft}

The following example illustrates why the max-left background traversal is
needed:
\begin{example}
  Consider the language $L_{aabb}$ on the alphabet $\Sigma=\{a,b,e\}$
  that has $e$ as a neutral letter
  and contains the words satisfying one of the following:
  \begin{enumerate}
  \item the word is exactly $aabb$ (up to the neutral letters),
  \item the word contains at least 2 $a$ and at least 3 $b$,
  \item the word contains at least 3 $a$ and at least 2 $b$.
  \end{enumerate}
  The language $L_{aabb}$ is in ZG (intuitively because the order between
  letters only matters when there is a small number of occurrences of each
  letter), and it is clearly extensible because every condition in the list
  above is extensible except the first one, however every extension of the words
  of the first condition is covered by one of the two other conditions.

  Hence, $L_{aabb}$ is in particular \threshext,
  and we can take threshold $p=3$.
  For \cref{threshdef-condp}, we have
  $\cond(\{a,b\})=\Sigma^*$,
  $\cond(\{a\})$ contains all words with at least two~$b$, and symmetrically
  $\cond(\{b\})$ contains all words with at least two~$a$.

  Consider now the following input word $w$
  without specification of the exact amount of neutral letters $e$.
  \begin{alignat*}{19}
    index:  &  & && i_1&& && i_2&& &&i_3&& &&i_4&&&& i_5&&&& i_6&&&& i_7&&&& i_8&&\\
    \boldsymbol{w} & \boldsymbol{=}&\quad e^*&&\quad \boldsymbol{a}&&\quad
    e^*&&\quad \boldsymbol{a}&&\quad e^*&&\quad \boldsymbol{a}&&\quad e^*&&\quad
    \boldsymbol{a}&&\quad e^*&&\quad \boldsymbol{b}&&\quad e^*&&\quad
    \boldsymbol{b}&&\quad e^*&&\quad \boldsymbol{b}&&\quad e^*&&\quad
    \boldsymbol{b}&&\quad e^*&&
  \end{alignat*}

  The algorithm will start the outer for-loop with $l=1$
  and then the inner while-loop with $r=|w|$.
  The right border can by decremented by only counting non-neutral letters.
  After position $r=i_7$, the letter $b$ becomes rare.
  The call to \textsc{GetRareOccurrencesSingle}($b$, 1, $i_7-1$) uses the min
  background traversal for~$b$, i.e., it calls
  \textsc{FinishMinBackgroundTraversal}($b$)
  to find the remaining positions $\Delta_b=\{i_5,i_6\}$.
  After position $r=i_6$, the amount of letters $b$ will be one
  and therefore for $r=i_6-1$, the infix $w[1,i_6-1]\not\in L_{aabb}$
  because $b\not\in \cond(\{a\})$.
  Thus, the first limit is $r_1=i_6$.
  At the end of the first iteration, we invoke
  \textsc{StartMaxLeftBackgroundTraversal}($a$, $\Delta_a$, $r_1$) for $a\in\{a,b\}$
  to start the max-left background traversal of every non-neutral letter.

  For $i_1+1\leq l\leq i_2$, the limits will be the same as $r_1$.
  These iterations are similar to the above
  except that \textsc{FinishMaxLeftBackgroundTraversal}($b$)
  will be called to compute $\Delta_b$.
  The max-left background traversal of $b$ can be finished in constant time at
  that point, 
  because it passed long enough over $\Lambda_b$ to find
  the $\leq p$ last $b$ left of $r_1$ ($\Delta_b=\{i_5, i_6\}$).
  The first time, these positions are saved, and then they are reused for further queries.

  At position $l=i_2+1$, directly at the beginning ($r=|w|$),
  letter $a$ is already rare.
  Up to this position, the max-left background traversal of $a$ could
  pass over $\Lambda_a$ and at most constantly many $a$ are left.
  So it can be finished
  and the result $\Delta_a=\{i_2,i_3,i_4\}$ can be stored for later use.
  After $r=i_7$, the letter $b$ becomes rare and $\Delta_b$ is filled.
  The inner while-loop is exited at $r=i_7-1$ because $T=\emptyset$.
  Now, to decide in \textsc{ProduceRemainingL}($i_2+1$, $i_7-1$, $\Delta$)
  if the infix $[i_2+1,i_7-1]$ (among others) is exactly $aabb$,
  we need the exact positions of the $\leq 2$ last $a$'s and $b$'s left of $r_1$
  that we find in $\Delta$ thanks to the max-left background traversals. In the
  word $w$ above, we must indeed produce the infixes $[i_2+1,i_7-1]$ (among
  others).

  Note that we never passed over the positions $i_3$ and $i_4$ of $a$ with our
  indices $l$ and $r$. So the input word would be indistinguishable from a word
  where $i_4$ contains $b$ and $i_5$ contains $a$ (in which case we should not
  produce the infix $[i_2+1,i_7-1]$), were it not for the max-left background
  traversals.
\end{example}

\end{toappendix}

\section{Discussion and Lower Bounds}
\label{sec:lower}

We have shown that \threshext languages $L$ enjoy dynamic
enumeration of $L$-infixes with linear preprocessing and memory, constant-time
updates,
constant delay, and constant enumeration memory. Languages with
such an algorithm are called \emph{tractable}.
In this section, we study which other languages are tractable and which are
(conditionally) not tractable.

For simplicity, in the rest of this section, we focus on languages featuring a
neutral letter, which we write $e$ and
omit from regular
expressions (e.g., $b^* a$ should be understood
as $(e+b)^*ae^*$). %
We first give prerequisites about the \emph{prefix-$U_1$ problem} used as a
hypothesis in almost all of our lower bounds. We then discuss the case of languages
outside of \ZG, and then discuss the case of
languages which are in \ZG but not \threshext.

\subparagraph*{Prefix-$\bm{U_1}$.}
Most lower bounds presented here are conditional to 
the \emph{prefix-$U_1$ hypothesis} from~\cite{amarilliICALP21}, that we now recall.
Let $n \in \NN$.
The \emph{prefix-$U_1$ problem} asks us 
for a data structure such that,
given a set $S \subseteq \{1, \ldots, n\}$,
we can perform the following tasks on input
$i\leq n$:
(1.)~test whether $i\in S$;
(2.) insert $i$ to $S$;
(3.) remove $i$ from $S$;
(4.) return \textsc{yes} if $i \leq \min S$.

The \emph{prefix-$U_1$ hypothesis} states that there is no data structure that can
achieve all these tasks in constant time (independent from $n$) in the unit cost
RAM model with logarithmic cell size. 
Note that the first three operations are easy to support in constant time (e.g.,
with a table of Booleans), so the challenge is to also support (4.).
The prefix-$U_1$ problem
is a weaker version of the well-known \emph{priority queue} problem (where 
we would replace 
(4.) by ``return $\min S$''). The priority queue problem
is itself a special case of the \emph{successor search} problem where 
(4.) would be: ``return $\min\{j \in S, i\leq j\}$''. There are data
structures solving successor search in time $O(\log\log n)$ and this
has been shown to be optimal (cf.~\cite{patrascu07} and
\cite[Appendix~A.2]{amarilliICALP21}).
The upper bound thus also covers the 
priority queue problem and prefix-$U_1$ problem as special cases; however, no lower bounds are known for these
problems specifically.

\begin{toappendix}
  \subsection{Languages Outside of ZG}
  \label{apx:nonzg}
\end{toappendix}

\subparagraph*{Languages outside of \ZG.}
It turns out that there are non-\ZG languages $L$ for which 
dynamic enumeration of $L$-infixes is tractable (which is not covered by our %
\cref{thm:main}), even though dynamic membership to~$L$ is not tractable by
\cref{prp:zgupdates}. For instance:

\begin{claimrep}
  \label{clm:l0}
  The language $L_0 \coloneq b^* a$ is not in \ZG but is tractable.
\end{claimrep}

\begin{proofsketch}
  We maintain an occurrence list $\Lambda_a$ of the $a$'s in the word under
  updates, and enumerate $L_0$-infixes by considering every $a$ and considering
  all ways to move the endpoints (over neutral letters, or over $b$'s for the
  left endpoint).
\end{proofsketch}

\begin{proof}
We maintain an occurrence list $\Lambda_a$ of the $a$'s in the word under
  updates. Then we can enumerate $L_0$-infixes 
by simply scanning $\Lambda_a$ to iterate over all positions $p$ of $w$
containing an $a$, and then expanding and
outputting the infixes around position $p$ until we reach an $a$ (or the
beginning of $w$) on the left endpoint or reach a non-neutral letter (or
the end of $w$) on the right endpoint.
\end{proof}

By contrast, the following similar language $L_1 \coloneq b^+ a$ is (conditionally)
intractable:

\begin{claimrep}
  \label{clm:l1}
  Under the prefix-$U_1$ hypothesis, the language $L_1 \coloneq b^+ a$ is not tractable
\end{claimrep}

\begin{proofsketch}
  We reduce from prefix-$U_1$: we store positions of a set $S$ as letter $b$'s,
  and insert a letter $a$ at a position~$i$ to test if the leftmost $b$ is $>i$
  or not.
\end{proofsketch}

\begin{proof}
We explain how to implement a prefix-$U_1$ data structure for a set $S$. Assume
  that all positions in $S$ carry an occurrence of $b$ and all positions not in~$S$ carry
an occurrence of the neutral letter $e$. For the last operation, on input $i$ we
  put an occurrence of~$a$
at position $i$: then there is an $L_1$ infix iff $i > \min S$, so enumerating
  $L_1$-infixes with constant delay gives an $O(1)$ algorithm for
prefix-$U_1$.
\end{proof}

To illustrate further how sensitive the problem is, consider the
following:

\begin{claimrep}
  \label{clm:l2}
  $L_2 \coloneq a\Sigma^*a$ is tractable, but $L_2'
  \coloneq b
  \Sigma^* a$ is not (under the prefix-$U_1$ hypothesis).
\end{claimrep}

\begin{proofsketch}
  For $L_2$ we consider all pairs of~$a$'s from the occurrence lists and
  consider all ways to move the endpoints left and right of the chosen $a$'s.
  For $L_2'$ we do like in \cref{clm:l1}.
\end{proofsketch}

\begin{proof}
  Enumerating $L_2$-infixes can be done tractably by
considering all pairs of positions of~$a$
using the occurrence lists and outputting the corresponding infixes (moving the
left and right endpoint respectively left and right of the chosen $a$'s, until
we reach another $a$ or the word boundaries).  By contrast, for $L_2'$ we can do
  the same reduction as for \cref{clm:l1}.
\end{proof}

We close the paragraph with two 
remarks. First, any \emph{extensible}
regular language outside of \ZG is not tractable under the
prefix-$U_1$ hypothesis. Indeed, we explained in \cref{sec:extensible} that, for extensible languages, 
dynamic membership for $L$ reduces to dynamic enumeration of $L$-infixes: thus
\cref{prp:zgupdates} concludes.
Second, some
regular languages outside of \ZG are \emph{unconditionally} not tractable, by
reducing from prefix-parity (see \cite{fredman1989cell}): see
Appendix~\ref{apx:nonzg}.

\begin{toappendix}
  We finish this section of the appendix by substantiating the claim made in the
  main text that one can show \emph{unconditional} lower bounds for some non-\ZG
  regular languages. It suffices for this to take any regular language $L$ which is extensible (so
  dynamic membership to $L$ reduces to $L$-infix enumeration, as discussed in
  \cref{sec:extensible}) and for which $L$ admits an unconditional lower bound
  on dynamic membership. This is the case, for instance, for languages outside
  of the class QSG studied in~\cite{amarilliICALP21}, which admit an
  unconditional lower bound of $\Omega(\log n / \log \log n)$ in the cell
  probe model by reduction from
  the so-called \emph{prefix-$\mathbb{Z}_d$} problem (using the result of
  \cite{fredman1989cell})). One concrete example of such a language, for
  instance, is the following regular language $L_\#$, on alphabet $\Sigma =
  \{\#, \$, a, e\}$ and with neutral letter~$e$:
  \[L_\# \coloneq \Sigma^*\# (aa)^* \$ a^* \#\Sigma^*\]
  Indeed it is clear how dynamic membership for this language, even on words
  that start and end with $\#$ and contain no other $\#$, admits a reduction
  from the prefix-parity problem. The same reduction shows that the language is
  not tractable, unconditionally, for the dynamic enumeration problem for
  $L_\#$-infixes.
\end{toappendix}

\begin{toappendix}
  \subsection{Languages in ZG}
\end{toappendix}

\subparagraph{Languages in \ZG.}
We now study the case of languages $L$ in \ZG and discuss whether
semi-extensibility is necessary for the language to be tractable. We first claim
that some non-semi-extensible \ZG languages are indeed (conditionally) not
tractable:

\begin{claimrep}
  \label{clm:l3}
  Consider the language $L_3$ on alphabet $\Sigma = \{a, b, c, e\}$
  (with $e$ neutral) that contains the words $w$ such that $|w|_c
  \geq 1$ and $w_{-\{c\}} \in L_{ab}$, for $L_{ab}$ the language of
  \cref{exa:threshext}. Then $L_3$ is non-semi-extensible ZG and is 
  not tractable under the prefix-$U_1$ hypothesis.
\end{claimrep}

\begin{proof}
  The language $L_3$ is in ZG (as can be shown, e.g., with the characterization
  of ZG languages given as~\cite[Corollary~3.5]{amarilli_paperman_ZG}). However
  it is not semi-extensible:
  in particular the words in $babc^+$ are not in~$L_3$ while those in
$abc^+$ are.

  If $L_3$ is tractable, then we can have an $O(1)$ data structure for the
  prefix-$U_1$ problem on a set~$S$, via the following reduction.
  For any position in $S$ we put 
a $c$ in our word and for any position outside of $S$ we have the neutral
  letter $e$. To perform the operation of point~(4.) on input $i$, we substitute the
positions $i$, $i+1$, and $i+2$ to write $bab$. (Note that up to padding the
word with neutral letters we can ensure that $i+2 \leq |w|$.)
Now, $w$ contains an $L_3$-infix iff $i < \max S$.
\end{proof}

However, some non-semi-extensible \ZG languages are tractable, for instance:

\begin{claimrep}
  \label{clm:l4}
  The language $L_4 \coloneq (aa)^*a$ of words with an odd number of $a$'s is tractable.
\end{claimrep}

\begin{proofsketch}
  We go over the occurrence list for~$a$ and enlarge $L_4$-infixes around
  each~$a$.
\end{proofsketch}

\begin{proof}
We go over
the occurrence list for~$a$ and enlarge infixes around every $a$, making sure
that we extend infixes either with neutral letters only, or with one $a$ to the
left and to the right simultaneously. Note that each infix is only produced
once, namely, when enlarging infixes around the $k$-th occurrence of $a$ that
  it contains, where
  $2k+1$ denotes the number of $a$'s in the factor realized by the infix.
\end{proof}

Thus, a characterization of the languages with tractable infix enumeration under
updates is still out of reach, even for ZG languages, indeed even for finite
languages with a neutral letter. For instance we do not know the status of $(aa)^+$, or even of $aa$.

\section{Conclusion and Future Work}
\label{sec:conc}
  \nosectionappendix

We have studied the dynamic enumeration of $L$-infixes for fixed 
languages
$L$, aiming at linear preprocessing, linear memory, constant-time updates,
constant delay, and constant-memory enumeration. Our main 
achievement is an algorithm that applies to any fixed \emph{semi-extensible \ZG
language}.
Further, we have %
discussed other tractable and intractable cases.

One obvious direction for further research is to achieve a complete
characterization of the tractable regular languages. %
However, the algorithm 
of \cref{thm:main}
shows that some tractable cases are technical, while other simple cases
(e.g., $L = aa$
with neutral letter) remain open. Some other tractable 
cases would deserve further study, e.g., languages like $a \Sigma^* b + b
\Sigma^* a$
(like \cref{clm:l2}), or languages generalizing the algorithm for
$(aa)^* a$ (cf \cref{clm:l4}).

Another research direction would be to characterize languages with different
complexity regimes below $O(\log n)$. In particular, one could aim at
generalizing the $O(\log \log n)$ class for dynamic membership from~\cite{amarilliICALP21}. The question can be generalized to more
general queries than $L$-infix queries, on words or even on
trees~\cite{amarilli2025dynamic}; or to more general update operations like
letter insertions or deletions. Another relevant type of queries, whose
complexity might differ, is \emph{direct access}
queries~\cite{bourhis2025dynamic}, where we
maintain under updates a data structure that allows us to directly access the
$i$-th $L$-infix of the word given the index~$i$; or \emph{tuple testing}
queries, where we can determine efficiently whether a given infix $[i,j]$
realizes a word of~$L$. More generally, a natural question for future work is
whether our results can generalize to more general queries in monadic
second-order logic (beyond $L$-infix queries).
It would also be interesting to study the alternative semantics where after each
update we must enumerate only the new results, though we expect that this would
require different techniques.

\begin{toappendix}
\vfill
\doclicenseThis
\end{toappendix}

\bibliography{bib}

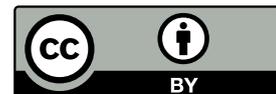
\end{document}